\documentclass[conference,10pt]{IEEEtran}
\ifCLASSINFOpdf
\else
\fi
\usepackage{amsmath}
\usepackage{hyperref}
\hypersetup{
    colorlinks = true,
    linkcolor = blue
    } 

  \usepackage{xcolor}
  \usepackage{amssymb}
  \usepackage[nameinlink]{cleveref}

  \usepackage{amsthm}
  \usepackage{thmtools}
  \usepackage{thm-restate}
  \usepackage{mathtools}
  \usepackage{enumitem}
  \usepackage{amsmath}

\usepackage{tkz-graph} 
\usetikzlibrary{arrows.meta}
\usepackage{xcolor}
\usepackage[normalem]{ulem}

  \declaretheorem[name=Theorem]{theorem}
  \declaretheorem[name=Proposition,numberlike=theorem]{proposition}
  \declaretheorem[name=Definition,numberlike=theorem]{definition}

  \declaretheorem[name=Lemma,numberlike=theorem]{lemma}

\newcommand{\ignore}[1]{}

  \DeclareMathOperator{\val}{val}

  \DeclareMathOperator{\OR}{OR}

  \newcommand{\Int}{{\mathbb Z}}

  \newcommand{\upRel}[2]{{#1}^{[#2]}}

\DeclareMathOperator{\arr}{\rightarrow}
\DeclareMathOperator{\larr}{\leftarrow}
\DeclareMathOperator{\sarr}{\leftrightarrow}
\newcommand{\relstr}[1]{\mathfrak{#1}}
\newcommand{\group}[1]{\mathcal{#1}}
\newcommand{\kpower}[1][k]{{#1} \arr  }

\newcommand{\plusarr}[1]{#1^{\arr}}
\newcommand{\minusarr}[1]{#1^{{\larr}}}
\newcommand{\orbiteq}{\sim}
\newcommand{\quotient}[1]{[#1]_{\orbiteq}}
\newcommand{\OG}{\relstr{U}}
\newcommand{\psarr}[1]{#1^{\sarr}}

\newcommand*\centermathcell[1]{\omit\hfil$\displaystyle#1$\hfil\ignorespaces}
\newcommand{\aut}{\operatorname{Aut}}

\newcommand{\pol}{\operatorname{Pol}}

\newcommand{\id}{\operatorname{id}}

\newcommand{\oc}{$\omega$-categorical }

\newcommand{\fa}{\mathfrak{A}}
\newcommand{\fb}{\mathfrak{B}}
\newcommand{\fc}{\mathfrak{C}}

\begin{document}
\title{Symmetries of structures\\
that fail to interpret something finite}%

\author{\IEEEauthorblockN{Libor Barto}
\IEEEauthorblockA{Charles University}
\and
\IEEEauthorblockN{Bertalan Bodor}
\IEEEauthorblockA{Charles University}
\and
\IEEEauthorblockN{Marcin Kozik}
\IEEEauthorblockA{Jagiellonian University}
\and
\IEEEauthorblockN{Antoine Mottet}
\IEEEauthorblockA{Technische Universität Hamburg}
\and 
\IEEEauthorblockN{Michael Pinsker}
\IEEEauthorblockA{Technische Universität Wien}
}

\maketitle

\thispagestyle{plain}
\pagestyle{plain}

\begin{abstract}
  We investigate structural implications arising from the condition that a given directed graph does not interpret, in the sense of primitive positive interpretation with parameters or orbits, every finite structure. Our results generalize several theorems   from the literature and yield further algebraic invariance properties that must be satisfied in every such graph. 
  Algebraic properties of this kind are tightly connected to the tractability of constraint satisfaction problems, and we obtain new such properties even for infinite countably categorical graphs. We balance these positive results by showing the existence of a countably categorical %
   hypergraph 
  that fails to interpret some finite structure, while still lacking some of the most essential %
  algebraic invariance properties known to hold for finite structures.
\end{abstract}

\IEEEpeerreviewmaketitle

\section{Introduction}

\subsection{The story}

  A major milestone in the theory of Constraint Satisfaction Problems (CSPs)  was a theorem due to Hell and Ne\v set\v ril~\cite{HellNesetril} which established a P/NP-complete complexity dichotomy for the computational problem of $H$-coloring of undirected graphs.
  Almost two decades later, Barto, Kozik and Niven~\cite{BartoKozikNiven} extended the dichotomy result to finite directed graphs with no sources and no sinks.
  Both results were subsumed by the CSP dichotomy theorem proven by Bulatov and Zhuk~\cite{Bulatov, ZhukFVConjecture}, which showed the same dichotomy holds for arbitrary finite directed graphs, or equivalently, arbitrary finite structures. 

  Both of the earlier results yield not only a complexity dichotomy, but a  structural dichotomy for the appropriate class of graphs,
  and derive the computational result as a direct consequence. 
  In the case of the theorem of Hell and Ne\v set\v ril, the structural result can be stated as follows: 
  any finite, undirected, loopless graph is either bipartite, in which  case the associated $H$-coloring problem is essentially just the $2$-coloring problem, or it \emph{pp-constructs} (in the sense of~\cite{wonderland}) the $3$-element clique $K_3$, and the $H$-coloring problem is as hard as $3$-coloring. If the graph is a core, i.e., the smallest template for its $H$-coloring problem, then it is either a single edge and defines precisely the $2$-coloring problem, or it \emph{pp-interprets} $K_3$, and consequently any finite structure, with parameters, providing a tangible witness of hardness. 
  
  By an observation of Siggers~\cite{Siggers}, the latter structural dichotomy for graphs also implies an algebraic invariance property for any finite core structure that fails to pp-interpret $K_3$ with parameters; similarly, the result of Barto, Kozik and Niven and its various extensions provide further such invariance properties, of which in particular the so-called weak near-unanimity (WNU) polymorphisms play a crucial role in Zhuk's proof of the CSP dichotomy theorem.

  While the proof of Hell and Ne\v set\v ril is purely combinatorial in nature,  the generalization of Barto, Kozik and Niven relies on the machinery developed in the context of the algebraic approach to CSP. The elegance  of this approach comes, so it seems,  at the cost of difficulties to generalize it to infinite countably categorical graphs. In fact, in the long line of research devoted to extending the CSP dichotomy theorem to the class
  of CSPs defined by first-order reducts of finitely bounded homogeneous structures (an important natural subclass of the class of $\omega$-categorical structures, see~\cite{Book, infinitesheep}), the only successful attempt so far to obtain structural dichotomies for $\omega$-categorical graphs similar to the ones mentioned above is due to Barto and Pinsker~\cite{Topo}. 
  Their approach builds on a streamlined  version, due to  Bulatov~\cite{BulatovHColoring},  of 
  the proof of Hell and Ne\v set\v ril, 
  and shows that any $\omega$-categorical graph which contains $K_3$ and which has no edge within an orbit of its automorphism group (a~\emph{pseudoloop}), pp-interprets together with the orbits of this group and parameters the clique $K_3$.  

  This result yields a non-trivial algebraic invariance property for any $\omega$-categorical structure which is a model-complete core and which fails to pp-interpret $K_3$ with parameters, and this property separates precisely tractable from intractable CSPs for in the realm of first-order reducts of finitely bounded homogeneous structures according to a conjecture of Bodirsky and Pinsker~\cite{BPP-projective-homomorphisms}. However, all attempts at a generalization of the more general, algebraic proof due to Barto, Kozik and Niven, not to mention the more advanced algebraic results available for general finite structures, 
  have failed. 
  The reason for that seems to be very elementary: 
  it is much easier to extend~%
   even very complicated 
  combinatorial constructions than to
  lift the algebraic notions tailored for finite algebras. As a result, no algebraic invariance properties except for the one obtained by Barto and Pinsker are known for the templates conjectured to have tractable CSPs.

In order to overcome this obstacle the following two-step approach is natural:%
  \begin{itemize}
      \item[(i)] provide combinatorial proofs of the known finite structural dichotomies; 
      \item[(ii)] lift the proofs to the countably categorical setting.
  \end{itemize}

  Upon closer inspection of step (ii),
  the following, obvious, difference between finite and infinite structures comes into focus.
  In the former, one can use \emph{all} elements of the structure as parameters in a single pp-interpretation, a standard trick applied throughout the entire theory. On the other hand, %
  countably categorical structures have finitely many tuples of fixed finite length only \emph{up to automorphisms}~%
  (i.e., the number of orbits of their automorphism group acting on tuples is finite). 
  Therefore generalizing results using orbits, rather than elements, as parameters seems more natural. 
  However, no results in this direction are known even for finite structures, adding another necessary step to the plan:   
    \begin{itemize}
      \item[(iii)] provide stronger finite dichotomies which use orbits of a permutation group rather than elements as parameters in an interpretation.
      \end{itemize}

  \subsection{Our contributions}
  We achieve combinatorial results in all three directions, and moreover contrast these with strong evidence that the algebraic methods from the finite do not lift to general countably categorical structures. 
  \subsubsection{Barto, Kozik and Niven, revisited}
    Our first contribution is a new, purely combinatorial proof of the dichotomy for finite digraphs with no sources and no sinks, 
     thus reproving the result of Barto, Kozik and Niven in the spirit of~(i) above. 
    In fact we provide two generalizations of the result, both for finite directed graphs. In our first generalization, we obtain a pp-interpretation with the orbits  of a subgroup of the automorphism group of the digraph, achieving~(iii) for the theorem of Barto, Kozik and Niven.
    \begin{theorem}%
    \label{thm:two}
Let $\relstr{A}$ be a finite smooth digraph, and let $\group G$ be a subgroup of its automorphism group $\aut(\relstr A)$. If $\relstr{A}$ is linked and without a loop, then $\relstr{A}$ expanded by  orbits of $\group G$ pp-interprets $K_3$ and hence EVERYTHING, i.e., every finite structure.%
\end{theorem}
As an immediate consequence of this theorem, we obtain for the first time a specific algebraic invariance property for any finite core structure which satisfies \emph{any} non-trivial algebraic invariance property, see~\Cref{thm:ugly_terms}. This result is a generalization of Sigger's result from~\cite{Siggers} which requires the structure to contain all singleton relations.

In the second generalization of Barto, Kozik, and Niven's result we are able to replace some of the assumptions on the digraph by ``pseudo-assumptions", i.e., assumptions on the graph after factorization by orbits of a subgroup of its automorphism group, making the proof amenable to the infinite setting.

    \begin{theorem}%
    \label{thm:one_a}
    Let $\relstr{A}$ be a finite smooth digraph, and let $\group G$ be a subgroup of its automorphism group $\aut(\relstr A)$.
    If the factor graph $\relstr{A}/\group G$ has algebraic length 1 and no loop, 
    then $\relstr{A}$  pp-interprets  with parameters EVERYTHING.
    \end{theorem}
    
    We remark that these two theorems are both consequences of a more general theorem we prove,~\Cref{thm:main_finite}. 

\subsubsection{Hell and Ne\v set\v ril, revisited}

    Our second main contribution is lifting~%
    the result of Hell and Ne\v set\v ril to the countably categorical case: 
    We show that if a graph without sources and sinks is such that
    the factor graph on the orbits of a subgroup of its automorphism group satisfies the assumptions in the theorem of Hell and Ne\v set\v ril, then the graph pp-interprets, using parameters and the  orbits of the group, all finite structures. This achieves~(ii) above for this theorem.
    
\begin{theorem}%
\label{thm:one_b}
    Let $\relstr{A}$ be a smooth digraph, and let $\group G$ be a subgroup of its automorphism group $\aut(\relstr A)$. Assume that $\group G$ has only  finitely many orbits in its action on pairs.  
    If $\relstr{A}/\group G$ is symmetric, loopless, and not bipartite, 
    then $\relstr{A}$ expanded by the orbits of $\group G$ pp-interprets  with parameters EVERYTHING. 
    \end{theorem}
    This result, applied to a finite graph and the trivial group $\group G$, provides a new and purely relational proof of the theorem of Hell and Ne\v set\v ril. It also implies a variety of algebraic invariance properties for any $\omega$-categorical structure that is a model-complete core and fails to pp-interpret $K_3$ with parameters, thereby answering an open problem in~\cite[Section 5.3]{Pseudo-loop}, vastly generalizing the result of Barto and Pinsker, and ending the long and dark period %
    of the uniqueness of their result -- see~\Cref{thm:pseudo_conditions}. %
\subsubsection{Mar\'{o}ti and McKenzie, revisited}
    In our third main contribution, we go beyond the realm of graphs and consider hypergraphs. The  above-mentioned algebraic invariance properties take the form of \emph{polymorphisms}, i.e.,  multivariate functions on the domain of the structure which leave the  structure invariant, satisfying non-trivial  \emph{identities}. The identities which have proven to be the most directly applicable to  questions of computational complexity of the corresponding CSP generally stem from hypergraphs rather than graphs; in particular, this is true for the weak near-unanimity (WNU) polymorphisms used in Zhuk's proof of the CSP dichotomy theorem. The latter were shown to exist for any finite core structure not pp-interpreting $K_3$ with parameters using the machinery of finite algebras by Mar\'{o}ti and McKenzie~\cite{MarotiMcKenzie}; further   algebraic proofs were given by Barto and Kozik~\cite{Cyclic} as well as Zhuk~\cite{Strong-Zhuk}.
    
    We show that similar polymorphisms need not exist for $\omega$-categorical structures under the same conditions. In fact, we obtain a much more general result,~\Cref{thm:counterexample_full}: any algebraic invariance property  which is a countably infinite disjunction of statements asserting the satisfaction of identities by polymorphisms %
     and with the property that no single member of the disjunction is implied for finite core structures which do not pp-interpret $K_3$ with parameters, can be avoided by an $\omega$-categorical hypergraph which does not pp-interpret $K_3$ with parameters. The existence of a weak near-unanimity (WNU) polymorphism is such a condition since although any finite core structure not pp-interpreting $K_3$ with parameters has a WNU polymorphism, the arity of the WNU polymorphism varies for different finite structures.
    
    Our theorem solves, in particular,~\cite[Problem 14.2.6 (21)]{Book}. It provides evidence that the full algebraic machinery available for finite structures does not lift to the $\omega$-categorical setting, and indicates that finer methods as presented e.g. in~\cite{SmoothApproximations} have to be developed for the narrower context of first-order reducts finitely bounded homogeneous structures in order to obtain the same algebraic invariance properties as in the finite.

\begin{theorem}%
\label{thm:three}
    There exists a hypergraph $\relstr{A}$ with the following properties:
    \begin{itemize}
        \item $\relstr{A}$  is  $\omega$-categorical; %
        \item $\relstr{A}$ has no pseudo-WNU polymorphisms; 
        \item $\relstr{A}$ expanded by the orbits of its automorphism group does NOT pp-interpret with parameters EVERYTHING.
    \end{itemize}
    \end{theorem}

\section{Preliminaries}
\label{sec:prelim}

\subsection{Relations}

Our terminology for relational structures is fairly standard and we often call them just structures. We %
  abuse notation by using the same name for a relational symbol and its interpretation in a structure;
   $R(a_1,\dotsc,a_n)$ typically means that $(a_1,\dotsc,a_n)$ is in the interpretation of $R$ in the given structure (which should always be clear from the context), while $R(x_1,\dotsc,x_n)$ is an atomic formula with free variables $x_1,\dotsc,x_n$.
A \emph{digraph} is a structure with a single binary relation that we usually denote by $\rightarrow$, its inverse
is denoted by $\leftarrow$. 
A \emph{graph} is a digraph that is symmetric, i.e., ${\rightarrow} = {\leftarrow}$.
A \emph{hypergraph} is a structure with a single relation. The induced substructure of $\relstr{A}$ on a subset $B$ is denoted $\relstr{A}|_B$, the quotient structure modulo an equivalence $\sim$ on the domain is denoted $\relstr{A}/{\!\sim}$. For instance, the quotient of a digraph $(A; \arr)$ is 
the digraph $(A/{\!\sim}, \{([a]_{\sim},[b]_{\sim}): a\rightarrow b\})$.

A first-order formula $\varphi$ is \emph{primitive positive} (pp, for short) if it consists of existential quantifiers, conjunctions, and atomic formulas only.
We say that $\relstr{A}$ pp-defines $\relstr{B}$ (or $\relstr{B}$ is pp-definable in $\relstr{A}$) if every relation in $\relstr{B}$ can be  defined by a primitive positive formula. The same terminology is used for sets of relations. 

A \emph{pp-interpretation} of a structure $\relstr B$ in $\relstr A$ consists of a partial surjective map $h\colon A^d\to B$ for some $d\geq 1$ such that for every $R\subseteq B^n$ that is either $B$, or the equality relation on $B$, or a relation of $\relstr B$, the preimage $h^{-1}(R)$ seen as a relation of arity $nd$ on $A$ has a pp-definition in $\relstr A$.
The integer $d$ is called the \emph{dimension} of the interpretation.

We say that $\relstr A$ pp-interprets/pp-defines $\relstr B$ \emph{with parameters} if $\relstr A$ expanded by  unary relations $\{a\}$ (where $a$ is in the domain of $\relstr A$)  pp-interprets/pp-defines $\relstr B$.
It is a classical fact that the 3-element clique $K_3$ pp-interprets every finite structure with parameters.

  An $n$-ary relation $R \subseteq A^n$ is \emph{subdirect} in $A$ if its projection to any coordinate is equal to $A$. 
  For two binary relations $R$ and $S$ on $A$, we write $R+S$ for the composition of $R$ and $S$ pp-defined by
  $$(R+S)(x,z) \equiv \exists y \ R(x,y) \wedge R(y,z).$$
  Accordingly, we write $nR$ for the $n$-fold composition of $R$ with itself.
  For a binary $R$ and unary $B$ on $A$ we also use $B+R$ pp-defined by 
  $$(B+R)(y) \equiv \exists x \ B(x) \wedge R(x,y).$$
  For a binary relation $\arr$ of a digraph, we use 
  $\plusarr{B}$ instead of $B + \arr$, for better readability.

  For two relations $R,S$ on $A$ with arities $n,m$ respectively we
  define an $(n+m)$-ary relation $\OR(R,S)$ by
  $\OR(R,S)(a_1,\dotsc,a_n,b_1,\dotsc,b_m)$
  if $R(a_1,\dotsc,a_n)$ or $S(b_1,\dotsc,b_m)$.
  Note that this \emph{is not} a pp-definition from $R$ and $S$.
  
  In our proofs of~\Cref{thm:two} and~\Cref{thm:one_a}, we will produce a pp-definition of $\OR(\alpha,\alpha)$ for a nontrivial equivalence relation $\alpha$. The following folklore observation (see \Cref{app:aux}) will finish the proofs. 

  \begin{restatable}{proposition}{ORprop}
  \label{prop:OR}
  Let $\relstr{A}$ be a finite core structure containing $\OR(\alpha,\alpha)$ for a proper equivalence relation $\alpha$ on some $B \subseteq A$. Then $\relstr{A}$ pp-interprets every finite structure.
  \end{restatable}

\subsection{Connectivity notions for digraphs} \label{subsec:prelim_digraphs}

Let $\relstr{A}=(A; \rightarrow)$ be a  digraph.
We say that $\relstr{A}$ is \emph{smooth} if $\arr$ is a subdirect relation on $A$, i.e., $\plusarr{A}=\minusarr{A}=A$. %

A \emph{walk}  in $\relstr{A}$ is a sequence $a_1 \ \epsilon_1 \ a_2\  \epsilon_2 \dots \epsilon_{n-1} \ a_n$, where $a_i \in A$ and each $\epsilon_i$ is either $\arr$ or $\larr$. The \emph{algebraic length} of such a walk is the number of forward arrows minus the number of backward arrows. We say that $\relstr{A}$ has \emph{algebraic length 1} if there exists a closed walk (i.e., $a_1=a_n$) of algebraic length~1.

A digraph is \emph{weakly connected} if there exists a walk between any two vertices. Weak components are defined accordingly as maximal induced subdigraphs that are weakly connected (or, abusing notation, the corresponding subsets of $A$).
  
  The $k$-fold composition of $\rightarrow$ with itself 
  is pp-defined by 
  \begin{equation*}
    x (\kpower) y \equiv \exists z_1,\dotsc,z_{k-1}\ x\rightarrow z_1 \rightarrow \dotsb \rightarrow z_{k-1}\rightarrow y.
  \end{equation*}
  Note that if $\rightarrow$ is subdirect in $A$, then so is $\kpower$ for any $k$.
  The \emph{link relation} for $\relstr{A}$ (or $\arr$) is defined as $L_{\rightarrow}(x,y) \equiv \exists z(x\rightarrow z \wedge y\rightarrow z)$. 
  Note that $L_{\rightarrow}$ is always symmetric, and if $\rightarrow$ is subdirect, then it is also reflexive.
  The transitive closure of $L_{\rightarrow}$ is called the \emph{linkness equivalence} associated with $\rightarrow$. 
  We call $\relstr{A}$ \emph{linked} if its linkness equivalence equals $A^2$.
  The \emph{$k$-link relation} for $\relstr{A}$ is $L_{(\kpower)}$, its transitive closure is the \emph{k-linkness equivalence}, and $\relstr{A}$ is \emph{$k$-linked} if $\kpower$ is linked.
  Since the link relation is reflexive for smooth $\relstr{A}$, its transitive closure can be pp-defined if $A$ is finite, namely by the formula 
  \begin{equation*}
    \exists z_1,\dotsc,z_{|A|}\ L_{\rightarrow}(x,z_1)\wedge L_{\rightarrow}(z_1,z_2) \wedge\dotsc\wedge L_{\rightarrow}(z_{|A|},y)\; .
  \end{equation*}
  The same formula then also works for $(\kpower)$ in place of $\rightarrow$. %
Note that a finite digraph $\relstr{A}$ is $k$-linked for some $k$ iff it is weakly connected and has algebraic length 1 (see, e.g., \cite[Claim 3.8]{Cyclic}). It also follows that a weak component of algebraic length 1 is pp-definable with parameters from $\arr$.
A graph is $k$-linked for some $k$ iff it is (weakly) connected and non-bipartite.

\subsection{Groups, orbits, $\omega$-categoricity, model-complete cores}
  
  Let $\group G$ be a permutation group acting on a set $A$. For $n\geq 1$, a equivalence relation ${\sim_{\group G}}$ on $A^n$ is defined by $x\sim_{\group G} y$ iff there exists an $\alpha\in \group G$ such that $\alpha(x)=y$. The equivalence classes of $\sim_{\group G}$ on $A^n$ are called \emph{$n$-orbits} of the group $\group G$. We say that $\group G$ is \emph{oligomorphic} if it has finitely many $n$-orbits for all $n\in \mathbb{N}$. A countable structure $\relstr A$ is called \emph{$\omega$-categorical}, or \emph{countably categorical}, if $\aut(\relstr A)$ is oligomorphic.
  
  For a digraph $\relstr A=(A,\rightarrow)$ and a group $\group G$ acting on $A$, we write $\relstr A/{\group G}$ instead of $\relstr{A}/{\sim_{\group G}}$ (where the equivalence is for $n=1$). 

  An $\omega$-categorical structure $\relstr A$ is a \emph{model-complete core} if for every endomorphism $e$ of $\relstr A$ and finite subset $F$ of $A$, there exists an automorphism $\alpha$ of $\relstr A$ such that $e$ and $\alpha$ coincide on $F$. A finite model-complete core is simply called a \emph{core}. Every $\omega$-categorical structure $\relstr A$ has an induced  substructure which is a model-complete core and which admits a homomorphism from $\relstr A$~\cite{Cores-journal}. This structure is unique up to isomorphism, and its isomorphism type is called the \emph{model-complete core of $\relstr A$}.

\subsection{Polymorphisms and identities}

A relation $R\subseteq A^n$ on $A$ is \emph{invariant} under an operation $f\colon A^m\to A$ if for all $r_1,\dots,r_m\in R$, the $n$-tuple $f(r_1,\dots,r_m)$ obtained by applying $f$ componentwise on $r_1,\dots,r_m$ is also in $R$. 
We say that an operation $f$ on $A$ is a \emph{polymorphism} of a structure $\relstr{A}$ with domain $A$ if each relation of $\relstr{A}$ is invariant under $f$. The set of all polymorphisms of $\relstr{A}$ is denoted $\pol(\relstr A)$.
Sets of polymorphisms are so-called \emph{clones}, and there is a tight connection between polymorphism clones of structures and their pp-definability and pp-interpretability strength; however, we do not expand on these in this paper and refer to~\cite{BoKaKoRo,CSPSurvey,BodirskyNesetrilJLC, Topo}. 

An \emph{equational condition} is a system of identities -- formal expressions of the form $s\approx t$, where $s$ and $t$ are terms over a common set of function symbols. 
We say that $\relstr{A}$ 
or  $\pol(\relstr{A})$ (or some set of operations on a common domain) satisfies an equational condition $\Sigma$, if the function symbols can be interpreted as members of $\pol(\relstr{A})$ so that, for each each $s \approx t$ in $\Sigma$, the equality $s = t$ holds for any evaluation of variables. 

An example is the 
\emph{weak near-unanimity (WNU) condition of arity $n$} given by $w(x,\ldots,x,y)\approx w(x,\ldots,x,y,x)\approx\dots\approx w(y,x,\ldots,x)$, for a symbol $w$ of arity $n$. It is satisfied in $\pol(\relstr{A})$ if $\relstr{A}$ has an $n$-ary polymorphism $w$  such that $w(x,\ldots,x,y) =  \dots\ = w(y,x,\ldots,x)$ for all $x,y$ in the domain; such a polymorphism is then called a WNU polymorphism. Similarly, the \emph{pseudo-WNU condition} is given by $u_1(w(x,\ldots,x,y))\approx u_2(w(x,\ldots,x,y,x))\approx\dots\approx u_n(w(y,x,\ldots,x))$, where the $u_i$ are unary symbols. Another example is the \emph{Siggers condition} $s(x,y,x,z,y,z) \approx s(y,x,z,x,z,y)$ and its pseudo-version
$u(s(x,y,x,z,y,z)) \approx v(s(y,x,z,x,z,y))$. 

An equational condition $\Sigma$ is called a \emph{minor condition} if for each identity $s \approx t$ in it, both $s$ and $t$ contain exactly one occurrence of a function symbol. For example, the WNU condition is minor, while pseudo-WNU is not. 
An equational condition is \emph{balanced} if in every $s \approx t$, 
  the same variables appear on the left- and right-hand side of the  identity; the examples above are such. 
 Finally, an equational condition is \emph{idempotent} if it entails $t(x,x, \dots, x) \approx x$ for every function symbol $t$. E.g., the \emph{idempotent WNU} condition $w(x,\ldots,x,y) =  \dots\ = w(y,x,\ldots,x), w(x,x, \ldots, x) \approx x$ is idempotent but not minor.

An equational condition is \emph{trivial} if it is satisfied in every polymorphism clone. %
 It is a folklore fact that if a structure pp-interprets $K_3$, then any equational condition it satisfies is trivial, and if it pp-interprets $K_3$ with parameters, then any idempotent condition it satisfies is trivial. %

\section{Loops and pseudoloops in finite digraphs} \label{sec:finite}
 
In this section we prove refined versions of~\Cref{thm:two,thm:one_a} and their consequences. Both theorems are derived from a single result, \Cref{thm:main_finite}, which we only state here. Its proof covers~\Cref{sec:main_finite_proof}.
  
  \subsection{Ranked relations and the main theorem}\label{subsec:ranked} %

  The crucial idea for the proof of the finite pseudoloop result stated in \Cref{thm:one_a} 
  is to change the digraph relation $\arr$, which has algebraic length 1 only modulo orbits, to a relation which truly has algebraic length 1 by appending the graph of an appropriate automorphism $g$. This new relation $\arr + g$ need not be pp-definable from the structure $\relstr{A}$, so pp-definability from $\arr + g$ does not have, in general, consequences on pp-definability from $\arr$. However, as we shall observe in~\Cref{lem:pp_from_rpp}, useful consequences can be recovered in case only restricted pp-formulas are used.
  This is the reason for the following definitions.
  
     A \emph{ranking} for a relation, or a relational symbol, $R$ of arity $n$ is a mapping $r:\{1,\dotsc,n\}\rightarrow\Int$. The pair $(R,r)$ is then called a 
     \emph{ranked relation}; for any coordinate $i\in \{1,\ldots,n\}$ (or variable,  in the case of a relational symbol), we call its value under $r$ the \emph{rank} of the coordinate.  A \emph{shift} of a ranked relation $(R,r)$  is any ranked relation $(R,r+k)$, where $k \in \Int$.
    \emph{Ranked relational} structures are defined accordingly.
    
    If $\rightarrow$ is a binary relation, then  $\rightarrow_{01}$ is the same relation together with the ranking which assigns rank $0$ to the first  and $1$ to the second coordinate.

    A \emph{ranked pp-formula}, \emph{rpp-formula} for short, over a set $\mathcal{R}$ of ranked relations is a pp-formula using ranked relations from $\mathcal{R}$,
    together with a ranking of the variables, which is a mapping $r:V\rightarrow\Int$ with the following properties:
    \begin{itemize}
        \item if $R(x_1,\dotsc,x_n)$ appears in the formula, then $R$ together with the ranking $i\mapsto r(x_i)$ is a shift of a ranked relation from $\mathcal{R}$;
        \item if $x=y$ appears in the formula, then $r$ is constant on $\{x,y\}$. %
    \end{itemize}
    Note that if the free variables of a ranked pp-formula are $y_1,\dotsc,y_n$,  then the formula defines a ranked relation  $R(y_1,\dotsc,y_n)$
    with ranking $i\mapsto r(y_i)$.

    If the rank of all the coordinates in a relation %
    is $0$, then %
    we call the relation \emph{$0$-ranked}. An rpp-formula is \emph{$0$-ranked} if all the free variables have rank $0$. %
    An important example of a $0$-ranked rpp-formula over $\arr_{01}$ is the formula defining $k$-linkness equivalence, and a rpp-formula over $\arr_{01}$ with parameters defining a weak component of algebraic length 1 discussed in \Cref{subsec:prelim_digraphs}. 
 
    Let $\relstr{A}$ and $\relstr{B}$ be ranked relational structures of the same signature, with domains $A$ and $B$, respectively. A
     sequence of functions $(f_i)_{i\in\Int}$, $f_i:A \to B$, 
    is a \emph{ranked homomorphism} from $\relstr{A}$ to $\relstr{B}$
    if for any shift $(R,r)$ of a ranked relation in $\relstr A$, say of arity $n$, we have that 
    $R(a_1,\dotsc,a_n)$ in $\relstr{A}$ implies $R(f_{r(1)}(a_1),\dotsc,f_{r(n)}(a_n))$ in $\relstr{B}$.
    Note that this property extends to rpp-definable relations; this is essentially the reason why the following lemma works. The proof is in~\Cref{app:aux}.

 \begin{restatable}{lemma}{PPfromRPP} \label{lem:pp_from_rpp}
Let $\relstr{A}=(A;\arr)$ be a digraph, $g \in \aut(\relstr{A})$,  and $S$ a relation on $A$. Let $\arr' = \arr + g$.
If $\arr'_{01}$  rpp-defines with parameters 0-ranked $S$, then $\arr$  pp-defines with parameters $S$. %
\end{restatable}

    Note that the set of ranked homomorphisms is closed under shifts, i.e., if $(f_i)_{i\in\Int}$ is a ranked homomorphism from $\relstr A$ to $\relstr B$, then  so is $(f_{i+k})_{i\in \Int}$, for any $k\in\mathbb Z$. Moreover, ranked homomorphisms are closed under composition, i.e., if $(g_i)_{i\in\Int}$ is a ranked homomorphism from $\relstr B$ to $\relstr C$, then  $(g_i\circ f_i)_{i\in\Int}$ is a ranked homomorphism from $\relstr A$ to $\relstr C$. In particular, the set of \emph{ranked automorphims} of $\relstr{A}$ (i.e., invertible ranked homomorphisms from $\relstr A$ to $\relstr{A}$) forms a group,
    the \emph{ranked automorphism group of $\relstr{A}$}.
    By the \emph{projection} %
    of a ranked automorphism group $\group H$, we mean the group 
     $\{f_0: (f_i)_{i\in\Int} \in \group H\}$.
     It is equal to $\{f_j: (f_i)_{i\in\Int} \in \group{H}\}$ for any $j$.

    We are ready to state the main result for finite digraphs.
    
    \begin{theorem}\label{thm:main_finite}
      Let $\relstr{A} = (A; \arr)$ be a finite smooth digraph, let $\group H$ be a subgroup of the ranked automorphism group of $(A; \rightarrow_{01})$, let $\group G$ be the  projection of $\group H$, and let $k\geq 1$. If $\arr$ is  $k$-linked and $(\kpower) \neq A^2$, then $\arr_{01}$ together with the $0$-ranked orbits %
      of $\group G$ rpp-define
      \begin{enumerate}
        \item  $B\varsubsetneq A$ such  that $\relstr{A}|_B$ is smooth  and $k$-linked, or  \label{main:red}
        \item 0-ranked $\OR(\alpha,\alpha)$
          for some proper equivalence relation $\alpha$ on some subset $C\subseteq A$. \label{main:proj}
      \end{enumerate}
    \end{theorem}

\subsection{Loops without parameters}

\noindent
The following refined version of \Cref{thm:two} is a simple consequence of \Cref{thm:main_finite}. For this result, the rankings are not needed.

    \begin{theorem}%
    \label{thm:two_refined}
      Let $\relstr{A}=(A; \arr)$ be a finite smooth digraph, and let $\group G$ be a subgroup of $\aut(\relstr A)$. If %
      $\relstr{A}$
      is linked, then $\rightarrow$ and the orbits of $\group G$ pp-define 
        \begin{itemize}
        \item some nonempty $B \subseteq A$ such that $B^2 \subseteq \arr$, or
        \item $\OR(\alpha,\alpha)$ for some proper equivalence relation  $\alpha$  on some subset $C \subseteq A$. 
        \end{itemize}
    \end{theorem}

    \begin{proof} 
       We apply~\Cref{thm:main_finite} with $k=1$, the same  $\relstr{A}$, and $\group H = \{ (g)_{i\in\Int}: g \in \group G\}$. Both the failure of the assumption $(\kpower) \neq A^2$ and \cref{main:proj} give the desired conclusion; if~\cref{main:red} holds, then we restrict $\relstr{A}$ 
       and $\group G$ to $B$ and apply~\Cref{thm:main_finite} again.
      In the end we either get $\OR(\alpha,\alpha)$ for a proper equivalence  relation $\alpha$ on some $C \subseteq A$, or we obtain a full subdigraph of $\relstr{A}$, as required.
    \end{proof}

Note that \Cref{thm:two} is an immediate consequence of the last theorem and \Cref{prop:OR}. 

A standard procedure for obtaining identities from structural results such as~\Cref{thm:two_refined} gives us the following corollary. The proof is given in~\Cref{apps:uglyterms} and the result further discussed in~\Cref{sec:conclusion}.

\begin{restatable}{theorem}{ThmUglyTerms} \label{thm:ugly_terms}
Let $\relstr{A}$ be a finite core structure and let $m$ be greater or equal to the number of elements of $\aut(\relstr A)$. Then the following are equivalent: 
\begin{itemize}
    \item $\relstr{A}$ does not pp-interpret all finite structures; in other words, $\relstr{A}$ has polymorphisms that satisfy \emph{some} nontrivial system of identities,
    \item $\pol(\relstr{A})$ satisfies 
\begin{alignat*}{8}
h(&&\alpha_1(x)&,\ldots,&\alpha_m(x)&,%
x,y,x,z,y,z)\approx\\
h(&&\centermathcell{y}&,\ldots,&\centermathcell{y}&,%
y,x,z,x,z,y).
\end{alignat*}
\end{itemize}
\end{restatable}

\subsection{Pseudoloops with pseudo assumptions}

    A refined version of~\Cref{thm:one_a} follows from~\Cref{thm:main_finite} by employing the trick mentioned in~\Cref{subsec:ranked}. We state a ``slighly infinite'' version that will be required in \Cref{sec:infinite}. We give a proof-sketch; the full proof is in~\Cref{apps:one_a_refined}.

    \begin{restatable}{theorem}{OneARefined} \label{thm:one_a_refined}
    Let $\relstr{A} = (A; \arr)$ be a smooth digraph and $\group G$ a subgroup of $\aut(\relstr A)$.
    If all weak components of $\relstr{A}$ are finite, and $\relstr{A}/\group G$ has algebraic length 1, 
    then 
    \begin{itemize}
        \item $\relstr{A}/\group{G}$ has a loop, or 
        \item 
        $\arr$ pp-defines with parameters $\OR(\alpha,\alpha)$ for some proper equivalence $\alpha$ on a finite $C \subseteq A$.
    \end{itemize}
    \end{restatable}

    \begin{proof}[Sketch of proof]
      Representatives of a  closed walk in $\relstr{A}/\group G$ of algebraic length 1 can be shifted using automorphisms from $\group G$ to a walk of algebraic length 1 from some $a$ to $b$ in the same orbit. We take $g \in \group G$ so that $g(b)=a$, define $\arr' = \arr + g$, and observe that the component of $a$ wrt. $\arr'$ has algebraic length 1 and is finite. Such components are known to be pp-definable with parameters. 
      We are in the position to keep applying~\Cref{thm:one_a_refined} as in the proof of~\Cref{thm:one_b_refined} with the caveat that the case $(\kpower) = A^2$ needs to be dealt with (but this is possible by a result from~\cite{Cyclic}).
      Application of   \Cref{lem:pp_from_rpp} finishes the proof. 
\end{proof}

    We provide two examples to illustrate the main theorems. The first one is an observation that has been proved and reproved repeatedly in the literature: the undirected 6-cycle is not invariant under any idempotent weak near-unanimity operation. Now this result follows from our general theorem:
    Take the automorphism $g$ of the 6-cycle according to the red arrow in~\Cref{fig:exampleone} and apply~\Cref{thm:one_a} to the graph and  $\group G = \{\mathrm{id},g\}$. We get that the graph with parameters pp-interprets every finite structure; the polymorphisms thus do not satisfy any nontrivial idempotent equational condition.

\begin{figure}
\centering
\begin{minipage}{.25\textwidth}
  \centering
        \begin{tikzpicture}[scale=0.7,rotate=90] 
          \Vertices[Lpos=90,unit=2]{circle}{0,5,4,3,2,1} 
          \Edges(1,2,3,4,5,0,1)
          \SetUpEdge[style={->,-stealth,bend left=20,ultra thick},color=red!40]
          \Edges(0,3,0)
          \Edges(1,4,1)
          \Edges(2,5,2)
        \end{tikzpicture}
  \caption{First example}
  \label{fig:exampleone}
\end{minipage}%
\begin{minipage}{.25\textwidth}
  \centering
        \begin{tikzpicture}[scale=0.8] 
         
          \Vertex{0} \Vertex[x=2,y=2]{1} \Vertex[x=4,y=0]{2} 
          \Edges(1,2,0,1))
        
          \Vertex[x=0,y=-1.4]{3} \Vertex[x=2,y=.6]{4} \Vertex[x=4,y=-1.4]{5} 
          \Edges(4,5,3,4)
        
          \SetUpEdge[style={->,-stealth,bend left,ultra thick},color=red!40]
          \Edges(0,3,0) \Edges(1,4,1) \Edges(2,5,2)
        \end{tikzpicture}
  \caption{Second example}
  \label{fig:exampletwo}
\end{minipage}
\end{figure}

    The second example shows that some of the assumptions in our results cannot be removed. We consider the graph $\relstr T_{3,3}$ in~\Cref{fig:exampletwo} and $\group{G} = \{\mathrm{id},g\}$, where $g$ is in red.
    On the one hand,~\Cref{thm:one_a} still shows that
    polymorphisms do not satisfy any nontrivial \emph{idempotent} identities. On the other hand, $\relstr T_{3,3}/\group G$ is linked and has no loop, but $\relstr T_{3,3}$ satisfies some non-trivial identities, see \cite[Example 6.3]{TopologyIsRelevant}.
    This shows, e.g., that one cannot switch in \Cref{thm:two} "$\relstr{A}$ linked" to "$\relstr{A}/\group G$ linked" or "$\relstr{A}$ symmetric non-bipartite",  and that parameters are necessary in \Cref{thm:one_a}.

\section{Proof of~\Cref{thm:main_finite}}

\label{sec:main_finite_proof}
  The entire section is devoted to the proof of~\Cref{thm:main_finite}.  We fix a finite smooth digraph $\relstr{A} = (A; \arr)$ and  a subgroup $\group H$ of the ranked automorphism group of $(A; \rightarrow_{01})$ whose projection is $\group G$. We assume that $\relstr{A}$ is  $k$-linked but  $(\kpower) \neq A^2$.

  As all the relations we work with in this section are on $A$,  we do not usually explicitly specify it, e.g., ``$R$ is subdirect'' means $R$ is subdirect in $A$. For convenience, we sometimes assume $A = \{1,2, \dots, |A|\}$.
  Let $O$ be the 0-ranked $\group G$-orbit of the tuple $(1,2, \dots, |A|)$. Note that for any $g:A \rightarrow A$, the tuple $(g(1), g(2), \dots, g(|A|))$ is in $O$ iff $g \in \group G$.

  We will, in three steps, show that $\arr_{01}$ and $O$ rpp-define a proper unary $B$ such that $\relstr{A}|_B$ is smooth and $k$-linked (\cref{main:red}) or 0-ranked $\OR(\alpha,\alpha)$ for a proper equivalence $\alpha$ on $C \subseteq A$ (\cref{main:proj}).

  \subsection{Constructing central or Q-central relations}

    We start with definitions, move on to a few auxiliary facts, and conclude the subsection with a first episode of the proof of~\Cref{thm:main_finite}. 

The concepts we now introduce play a significant role in Rosenberg's classification of maximal clones~\cite{Ros70}; we take them from Pinsker's presentation in~\cite{Pin02} of Quackenbush's proof of the classification~\cite{Qua71}. %
A relation $R$ is \emph{totally symmetric} if
  for every $(a_1,\dotsc,a_n)\in R$ and every permutation $\sigma$ on $\{1, 2, \dots, n\}$, the tuple 
  $(a_{\sigma(1)},\dotsc,a_{\sigma(n)})$ is in $R$.
  Similarly, a relation $R$ is \emph{totally reflexive} if
  any tuple from $A^n$ with at least two repeating entries belongs to $R$.
  A relation which is totally symmetric and totally reflexive will be called a \emph{TSR-relation}.

  The next concept, a center, is additionally very useful in the  theory of CSPs, e.g., in Zhuk's dichotomy proof~\cite{ZhukFVConjecture} or in absorption theory~\cite{absorption} (cf.~\cite{MinimalTaylor}). We remark that in some of the mentioned literature, the terminology slightly differs.
    
    \begin{definition}[center]
      Let $n\geq 2$. We call a relation $R\subseteq A^n$ \emph{P-central}~%
       if it is subdirect and
      the set
      \begin{equation*}
        \{a\in A: \forall a_2,\dotsc,a_{n}\ (a,a_2,\dotsc,a_{n})\in R\}
      \end{equation*}
      is nonempty. 
      In such a case, the above set is called a \emph{P-center}. The P stands for ``power"; for $n=2$ we call a P-central relation \emph{central} and P-center a \emph{center}. 
    \end{definition}
    
    \begin{definition}[central equivalence]\label{def:2center}
      A relation $R\subseteq A^n$~($n>2$) is \emph{PQ-central}~%
      (eQuivalence-central)
      if its projection to any two coordinates is full and the binary relation
      \begin{equation*}
        \alpha=\{(a,a')\in A^2: \forall a_3,\dotsc,a_n\ (a,a',a_3,\dotsc,a_{n})\in R\}
      \end{equation*}
      is an equivalence relation.
      The above equivalence relation is then called \emph{P-central}.
      For $n=3$ we talk about Q-central relations and central equivalence relations.
    \end{definition}

    The first fact is an easy observation whose proof is postponed to~\Cref{Pinsker}.
    
    \begin{restatable}{lemma}{gettingTSR}\label{prop:gettingTSR}
      Every subdirect and linked but not central ranked relation 
      rpp-defines a proper $0$-ranked TSR relation. 
      Moreover, if the new relation is binary, then it is additionally linked. %
    \end{restatable}
    
    We will use the following fact stated in~\cite{Pin02}; 
    we provide the proof in~\Cref{Pinsker}. In this lemma we need not worry about the ranking of the relations,
    as all the applications will use $0$-ranked relations exclusively (so pp-definitions will automatically give rise to $0$-ranked rpp-definitions).
    
    \begin{restatable}{lemma}{PorPQ}\label{prop:Rosenberg}
      Each proper TSR relation of arity at least 3,
      and each proper linked TSR relation of arity 2
      pp-defines a proper TSR relation which is P-central or PQ-central.
    \end{restatable}

    The next order of business is to get rid of the powers.
    This can be achieved by using any relation containing only ``surjective'' tuples, such as $O$.
    \begin{lemma}\label{prop:centerispp}
      A ranked TSR P-central relation $R$ and $O$ rpp-define the center of $R$ and a TSR central (binary) $R'$ with the same center (the ranking of $R'$ is inherited from the first two coordinates of $R$).
    \end{lemma}
    
    \begin{proof}
      Let $n$ be the arity of $R$. 
      The relation
      \begin{multline*}
        R^{n-1}(x_1,\dotsc,x_{n-1}) = \exists y_1,\dotsc, y_{|A|}\\ O(y_1,\dotsc,y_{|A|}) \wedge \bigwedge_i R(x_1,\dotsc,x_{n-1}, y_i)
      \end{multline*}
      is a ranked TSR P-central relation of arity $n-1$ with the same center as $R$ whenever $n \geq 3$.
      Repeating the construction~%
      (substituting $R^{n-1}$ for $R$)
      we can get down to $R^2$, which can be taken for $R'$, 
      and to $R^1$, which is an rpp-definition of the center.
    \end{proof}
 
    A similar result holds for PQ-central relations. In this case we again need not worry about the ranking of the relations. The proof is almost identical to the proof of~\Cref{prop:centerispp} and we skip it.

    \begin{lemma}\label{prop:PQtoQ}
      A TSR PQ-central relation $R$ and $O$
      pp-define the central equivalence of $R$ and a TSR Q-central (ternary) relation $R'$ with the same central equivalence.      
    \end{lemma}

    We are ready to proceed with the first step of the proof of~\Cref{thm:main_finite}.

   \begin{lemma}[First step]\label{lem:first_step}
   The ranked relations $\arr_{01}$ and $O$ rpp-define a proper ranked central or a proper 0-ranked Q-central relation.
   \end{lemma}
    
    \begin{proof}
      If $\kpower$ is central, we've already accomplished our goal.
      Otherwise we apply~\Cref{prop:gettingTSR} to obtain a TSR relation,
      then~\Cref{prop:Rosenberg} to obtain a $0$-ranked TSR relation which is P-central or PQ-central.
      We finish by applying~\Cref{prop:centerispp} or~\Cref{prop:PQtoQ}~%
      (depending on the case we are in)
      to end up with the required central or 0-ranked Q-central relation.
    \end{proof}

  \subsection{Constructing an OR relation}
    The plan for this subsection is to either rpp-define 0-ranked $\OR(T,T)$ for a proper TSR $T$, or $B$ as in \cref{main:red} of~\Cref{thm:main_finite}. 
    The proof is split in two parts depending on the relation obtained from the first step. 
    The case of 0-ranked Q-central relation is dealt with in the following lemma (where, again, we need not care about rankings).

      \begin{lemma}[Second step, Q-central case]\label{lem:orfromQ}
        A proper Q-central relation $R$ invariant under $\group G$ and $O$ pp-define 
        $\OR(T,T)$ for a proper TSR $T$.
      \end{lemma}
      
      \begin{proof} 
        We assume that $A=\{1,\dotsc,n\}$, let $R$ and $O$ are as in the statement.
        Denote the central equivalence of $R$ by $\alpha$~%
        and note that by~\Cref{prop:PQtoQ} $\alpha$ is pp-definable in $R$ and $O$.

        We choose $a,b\in A$ such that
        \begin{itemize}
          \item $(a,b)\notin \alpha$ and
          \item the set $I=\{i:(a,b,i)\in R\}$ is maximal~(under inclusion) among similar sets defined for other $(a,b)\notin\alpha$.
        \end{itemize}
        We assume that $a=1$ and $b=2$~%
        (this can be obtained by renaming the elements of $A$).

        Next, we find a minimal number $k$ such that
        \begin{itemize}
          \item every $(k-1)$-element subset of $A$ is included in $g(I)$ for some $g \in \group G$ and
          \item some $k$-element subset of $A$ is \emph{not included} in $g(I)$ for any $g \in \group G$.
        \end{itemize}
        We can always find such a $k$ and it will satisfy $1\leq k \leq |I|+1\leq n$;
        this is a consequence of $(a,b)\notin\alpha$.

        Our $T$ will consists of tuples $(c_1,\dotsc,c_k)$ such that
        $\{c_1,\dotsc,c_k\}\subseteq g(I)$ for some $g$ from $\group{G}$. %
        Note that  $T$ is totally symmetric by definition and totally reflexive by the choice of $k$.
        The choice of $k$ ensures, at the same time, that $T\neq A^{k}$.

        Our formula will have free variables $x_1^1,\dotsc,x_k^1$
        and $x_1^2,\dotsc,x_k^2$.
        The formula is
        \begin{align*}
          &\exists y_0,y_1,y_2,z_1,\dotsc,z_n\\  
          &y_0 = z_1 \wedge y_2 = z_2 \wedge
          O(z_1,\dotsc,z_n) \wedge \\
          &\bigwedge_{j=1}^2 \biggl(\bigwedge_{i=1}^k R(y_{j-1},y_j,x_i^j) \wedge 
          \bigwedge_{i\in I} R(y_{j-1},y_j,z_i)\biggr) \enspace.
        \end{align*}

        It remains to verify that the formula works.
        First take $(a_1,\dotsc,a_{k},b_1,\dotsc,b_k)\in\OR(T,T)$.
        Say that $\{a_{1},\dotsc,a_k\} \subseteq g(I)$ for some $g \in \group G$~%
        (the case $\{b_1,\dotsc,b_k\} \subseteq g(I)$ is symmetric).
        We choose a witnessing evaluation of quantified variables as follows: $y_0 \mapsto g(1)$, $y_1,y_2 \mapsto g(2)$,
        and $z_i$ to $g(i)$.

        The first three (simple) conjuncts hold by construction.
        Let's focus on the complex conjunct. 
        For $j=2$, the first two arguments of $R$ are identical, thus in $\alpha$, which is central and therefore $x_i^2$'s can be arbitrary.
        For $j=1$, we recall that $R(1,2,i)$ for every $i\in I$ and, since $g$ is an automorphism, we get  $R(g(1),g(2),g(i))$ as required.

        For the opposite direction, let $\val$ be an evaluation of variables $x_i^j, y_j, z_i$ making the quantifier-free part true. 
        The third conjunct ensures that $g: i\mapsto \val(z_i)$ is from $\group G$. We can therefore
        define a new valid evaluation of variables by $\val'(x) = g^{-1}(\val(x))$.
        The new evaluation satisfies $\val'(z_i)=i$, $\val'(y_0)=1$ and $\val'(y_2)=2$.
        If, for $j=1$ or $j=2$, $\{\val'(x_1^j),\dotsc,\val'(x_k^j)\}\subseteq I$ we achieved our goal.

        Let $j$ be such that $(\val'(y_{j-1}),\val'(y_j))\notin\alpha$ and note that 
        $R(\val'(y_{j-1}),\val'(y_j),\val'(z_i)=i)$ holds for every $i\in I$.
        But then the existence of $i$ with $\val'(x_j^i)\notin I$ would contradict the maximality of $I$~%
        (as the formula ensures that $R(\val'(y_{j-1}),\val'(y_j),\val'(x^j_i))$.
        Therefore no such $i$ exists and the proof is concluded. 
      \end{proof}

    The case of ranked central relation is more complex. 
       We start with a ranked central relation, and use it, together with $\rightarrow_{01}$, to construct
      another a subset $B$ with $\relstr{A}|_B$ smooth; such a construction appears in, e.g., \cite{Cyclic}. A proof is provided in~\Cref{apps:step_two_b}. 

      \begin{restatable}{lemma}{LemmaWalking}\label{lem:walking}
      A nonempty $C \varsubsetneq A$ and $\arr_{01}$ rpp-define  a nonempty  $B \varsubsetneq A$ such that $\relstr{A}|_B$ is smooth.      
      \end{restatable}

      We remark that for any rpp-formula $\phi$ in a unary $C$ and $\arr_{01}$, the formula obtained by removing all conjuncts $C(x)$ defines $A$. This is because $\arr$ is subdirect: witnesses for quantified variables can be obtained from infinite walks to and from a given vertex (evaluate all variables of rank $r$ to the $r$th vertex in this bi-infinite walk).

      \begin{restatable}[Second step, central case]{lemma}{StepTwoB}\label{lem:orfromcenter}
      Let $R$ be a ranked central relation invariant under $\group{H}$. Then $R$, $O$, and $\arr_{01}$ rpp-define
      \begin{itemize}
        \item $\emptyset \neq B\varsubsetneq A$ such that $\relstr{A}|_B$ is smooth  and $k$-linked, or  
        \item 0-ranked $\OR(T,T)$ for some proper TSR $T$.
    \end{itemize}
      \end{restatable}
      \begin{proof}[Proof beginning] 
        Let $C$ be the center of $R$ and $\psi$ be the rpp-formula defining $B$ from \Cref{lem:walking}, i.e., $B$ is a proper nonempty subset of $A$ and $\relstr{A}|_B$ is smooth.
        We assume that $\relstr{A}|_B$ is not $k$-linked; let $\alpha$ be the $k$-linkness equivalence relation on $B$. Our aim is to rpp-define $\OR(T,T)$ for a proper TSR $T$. 
        
        Take a 0-ranked rpp-formula $\varphi$ with two free variables defining the $k$-linkness relation~%
        (with $|A|$-many links). %
        Since $\relstr{A}$ is linked, it defines $A^2$. 
        Let $\varphi'$  be obtained from $\varphi$ by adding a conjunct $B(x)$ for every variable~%
        (i.e., both the quantified and the free variables).
        Now $\varphi'$ means ``being $k$-linked in $\rightarrow$ restricted to $B$'', so it is a $0$-ranked pp-definition of $\alpha$.

        Next, we define $\varphi''$ by replacing in $\varphi'$  each conjunct $B(x)$ on a quantified variable $x$ by the formula $\psi(x)$. Clearly, the formula $\varphi''$ still defines  $\alpha$. Also note that if we remove all the conjuncts $C(x)$ (they all come from quantified variables), then all the restrictions on the original quantified variables of $\varphi'$ are dropped (see the remark before the lemma), so the obtained formula defines $B^2$. 
        
        In $\varphi''$ we remove, one by one, the conjuncts $C(x)$. %
        At some point we arrive to a formula with a selected, quantified variable $x$ that defines a subset of $B^2$ strictly larger than $\alpha$, but if we added back the conjunct $C(x)$, it would define $\alpha$. By making $x$ free, we get an rpp-definition (using $B$, $C$, $\arr_{01}$) of a ternary relation $S$ such that   

        \begin{itemize}
          \item $\exists x\ S(y,y',x)$ is a subset of $B^2$ %
          larger than $\alpha$, and
          \item $\exists x\ S(y,y',x)\wedge C(x)$ is $\alpha$.
        \end{itemize}
        With this ternary relation in hand,
        the remaining reasoning is somewhat similar to  what was done in \Cref{lem:orfromQ}. 
        We provide the full proof in~\Cref{apps:step_two_b}
    \end{proof}

    \subsection{Third step}

      \noindent A plan for this section is as follows.
      We start with $0$-ranked $\OR(T,T)$ for a TSR relation $T$,
      next we improve $T$ to a P-central or PQ-central TSR,
      and then get $\OR(C,C)$ for $C\subseteq A$ or $\OR(\alpha,\alpha)$ for an equivalence on $A$.
      The second case is already~\cref{main:proj} of~\Cref{thm:main_finite}.
      In the first case we need to still work with $\rightarrow_{01}$ to end up in~\cref{main:red} or in~\cref{main:proj}. %
      Until the last lemma, all the relations and formulas  are 0-ranked.

      The first lemma is proved in~\Cref{apps:step_three}. %
      
      \begin{restatable}{lemma}{GettingPorPQor}\label{prop:gettingPorPQor}
        Let $T$ be a proper TSR relation.
        The relation $\OR(T,T)$ pp-defines $\OR(S,S)$ with a nonempty proper $S$ such that 
        $S$ is unary, or $S$ is an equivalence on $A$, or $S$ is TSR and P-central,
        or $S$ is TSR and PQ-central.
      \end{restatable}

      \begin{lemma}\label{prop:noP}
        Let $R$ be  %
        PQ-central with P-central equivalence $\alpha$.
        Then $\OR(R,R)$ and $O$ pp-define $\OR(\alpha,\alpha)$.

        Let $R$ be %
        P-central with P-center $C$. %
        Then $\OR(R,R)$ and $O$ pp-define $\OR(C,C)$.
      \end{lemma}
      \begin{proof}
        We deal with the PQ-central case only, 
        the proof for the P-central relations is analogous.
        Let $n>2$ be the arity of $R$. We define
        \begin{multline*}
          \OR(\alpha,R)(x_1,x_2,y_1,\dotsc,y_n) \equiv \\
          \exists z_1,\dotsc,z_{|A|}\ 
          O(z_1,\dotsc,z_{|A|})\ \wedge \\
          \bigwedge_{f:\{3,\dotsc,n\}\rightarrow |A|} \OR(R,R)(x_1,x_2, z_{f(3)},\dotsc,z_{f(n)},y_1,\dotsc,y_n)
        \end{multline*}
        We repeat the construction
        with $\OR(R,\alpha)$ in place of $\OR(R,R)$ and appropriately permuted coordinates, to obtain $\OR(\alpha,\alpha)$.
      \end{proof}

   The final proof, given in~\Cref{apps:step_three},  begins with a similar process to the proof of \Cref{lem:orfromcenter} to obtain a suitable ternary relation, which in combination with $\OR(C,C)$ quite easily gives $\OR(\alpha,\alpha)$.

      \begin{restatable}{lemma}{FinalOR}\label{prop:finalOR}
      Let $C$ be a nonempty proper unary relation. %
      Then 0-ranked $\OR(C,C)$ and $\arr_{01}$ rpp-define
      \begin{itemize}
        \item $\emptyset \neq B\varsubsetneq A$ such that $\relstr{A}|_B$ is smooth  and $k$-linked, or  
        \item 0-ranked $\OR(\alpha,\alpha)$ for some proper equivalence $\alpha$ on $B \subseteq A$.
    \end{itemize}
      \end{restatable}

This lemma finishes the proof of \Cref{thm:main_finite}. To summarize,
\Cref{lem:first_step} together with \Cref{lem:orfromQ} or \Cref{lem:orfromcenter}  produces, unless we are already done, a 0-ranked $\OR(T,T)$ for some proper TSR $T$. \Cref{prop:gettingPorPQor} improves $T$ and \Cref{prop:noP} further improves it, when necessary, to a nontrivial unary $T$ or $T = \alpha$ which is an equivalence. In the latter case, we are done, and in the former one, \Cref{prop:finalOR} finishes the proof.

\section{Pseudoloops in infinite digraphs} \label{sec:infinite}

\noindent In this section we state and prove~\Cref{thm:one_b_refined}~(below) which is a refined version of \Cref{thm:one_b}.

\begin{restatable}{theorem}{OneBRefined} \label{thm:one_b_refined}
Let $\relstr{A} = (A,\arr)$ be a smooth digraph and $\group G$ a subgroup of $\aut(\relstr A)$. If $\relstr{A}/\group{G}$ is symmetric, non-bipartite, and $\group{G}$ has finitely many 2-orbits, then
\begin{itemize}
\item $\relstr{A}/\group{G}$ has a loop, or
\item $\arr$ and 1-orbits of $\group{G}$ 1-dimensionally pp-interpret with parameters $\OR(\alpha,\alpha)$ for some proper equivalence $\alpha$ on a finite $C$.
\end{itemize}
\end{restatable}

Just like~\Cref{thm:one_a_refined}, 
the last theorem also has a consequence on identities, proved in \Cref{apps:pseudo_conditions}. Given a finite digraph $\relstr{B} = (\{1, 2, \dots, n\}; \{(i_1,j_1), (i_2,j_2), \dots, (i_m,j_m)\})$,
the \emph{$\relstr{B}$-pseudoloop condition} is
$u(s(x_{i_1}, \dots, x_{i_m})) \approx v(s(x_{j_1}, \dots, x_{j_m}))$. E.g., pseudo-Siggers condition is the $K_3$-pseudoloop condition.

\begin{restatable}{theorem}{ThmPseudoConditions}\label{thm:pseudo_conditions}
  Let $\relstr{A}$ be an $\omega$-categorical model-complete core structure. Then the following are equivalent.
\begin{enumerate}
\item $\relstr{A}$ \emph{does not} pp-interpret with parameters all finite structures;
\item $\pol(\relstr{A})$ satisfies some nontrivial pseudoloop condition; 
\item $\pol(\relstr{A})$ satisfies $\relstr{B}$-pseudoloop condition for every symmetric non-bipartite  $\relstr{B}$. 
\end{enumerate} 
\end{restatable}

\subsection{Triangle configurations}

\noindent
In this section we prove \Cref{thm:one_b_refined} by reducing it to \Cref{thm:one_a_refined}.
Our proof~%
(details in~\Cref{apps:one_b_refined})
uses a key concept of triangle configuration~%
(defined below) and splits into two parts:
in part one we prove~\Cref{thm:triangle_config}
and in part two we adapt the proof of Proposition 3.8 in~\cite{Topo} to conclude the reasoning.
This section contains part one:
we work under the assumptions of \Cref{thm:one_b_refined}, 
to pp-define a special configuration of three subsets, which we now introduce.

A \emph{triangle configuration} for $\relstr{A}=(A;\arr)$ and $\group G$ is a quadruple $(P,P_0,P_1,P_2)$ of subsets of $A$ such that
\begin{enumerate}
\item[(i)] %
$P = \bigcup P_i$ and $P_0$, $P_1$, and $P_2$ are pairwise disjoint.
\item[(ii)] $\plusarr{P_i}\cap P_i=\emptyset$ for $i=0,1,2$, i.e., $P_i$ are independent.
\item[(iii)] $\plusarr{P_i} \supseteq  P_j$ and $\minusarr{P_i} \supseteq P_j$  for all $i,j=0,1,2$ with $i \neq j$.
\item[(iv)] $\relstr{A}|_P/\group{G}$ is non-bipartite  
\end{enumerate}
We say a set of relations pp-defines a triangle configuration if it pp-defines all the four subsets.
The next section is devoted to a proof of the following Proposition:

\begin{proposition} \label{thm:triangle_config}
	Let $\relstr A = (A; \arr)$ be a smooth digraph and $\group G$ a subgroup of $\aut(\relstr A)$.  
	If $\relstr{A}/\group{G}$ is symmetric, without a loop, and non-bipartite,  and $\group G$ has finitely many 1-orbits, 
	then $\arr$ and $1$-orbits of $\group G$ pp-define, for some $k$, a triangle configuration for $(A;\kpower)$ and $\group G$.
\end{proposition}

\subsection{Proof of~\Cref{thm:triangle_config}}
\label{subsec:triangle_config}

  \noindent We fix a smooth digraph $\relstr A = (A; \arr)$ and a subgroup $\group G$ of  $\aut(\relstr A)$ 
  satisfying assumptions of~\Cref{thm:triangle_config}. 
  Most of the work takes place on the 1-orbit graph $\OG = \relstr{A}/\group{G}$.
  It is, by the assumptions, indeed a graph, which is finite, non-bipartite and without a loop.   The domain of $\OG$, denoted $U$, is the set of 1-orbits. 
We use $\sarr$ to denote its relation, and vertices $u \sarr v$ are called \emph{adjacent} or \emph{neighbors}.

  The goal is to prove that $\arr$ and  1-orbits of $\group G$ pp-define a triangle configuration for $(A, \kpower)$ and $\group G$ for some $k$. The proof is by contradiction; we assume that $\relstr{A}$ and $\group G$ is a counterexample with the smallest possible $|U|$.
 
  The strategy of the proof is to pp-define a triangle configuration in the orbit graph $\OG$. In order to be able to lift pp-definability back, we need to restrict the allowed pp-definitions as follows.
  
\begin{definition}\label{ppt_alter}
	A subset $S\subseteq U$ is \emph{tree-definable}%
	\footnote{The term stems from the fact that tree-definability is the same as definability with parameters by a tree formula, but this is irrelevant for this paper.}
	if $S$ is contained in the smallest family $\mathcal{S}$ of subsets of $U$ such that
\begin{enumerate}
\item $\mathcal{S}$ contains $U$ and all singletons, and
\item if $B_1,B_2\in \mathcal{S}$, then $\psarr{B_1} \in \mathcal{S}$ and $B_1\cap B_2\in \mathcal{S}$.
\end{enumerate}
\end{definition}

\noindent
	It is easy to prove by induction that  if $S\subseteq U$ is tree-definable, then $\bigcup{S} \subseteq A$ is pp-definable from $\arr$ and 1-orbits of $\group G$. 
  
 As the first step, we replace $\arr$ by $\kpower[(k-2)]$, where $k$ is the length of the shortest odd cycle in $\OG$. This new digraph $\relstr{A}$ is still  a counterexample to the proposition, since the new $\OG$ is  non-bipartite and without a loop. Additionally, it contains a triangle.
 
 As the next step, we derive some consequences of the minimality of our counterexample.

\begin{lemma}\label{lem:tree_def_bipartite}
If $S \varsubsetneq U$ is tree-definable, then $\OG|_S$ is bipartite.
\end{lemma}

\begin{proof}
	Suppose $\OG|_S$ is not bipartite. Since $S$ is tree-definable, the set $B = \bigcup S \varsubsetneq A$ is pp-definable from $\arr$ and 1-orbits of $\group{G}$, in particular, it is a union of 1-orbits of $\group{G}$. Moreover, the digraph $\relstr{A}|_S$ together with $\group{G}|_S$ satisfies the assumptions of~\Cref{thm:triangle_config} and it has strictly smaller number of orbits. Since our counterexample is minimal, we obtain a pp-definable triangle configuration  for this restricted digraph, which is a triangle configuration for $\relstr{A}$, a contradiction.
\end{proof}

\begin{lemma} \label{lem:common_neighbors}
 Any two vertices of $U$ have a common neighbor. %
\end{lemma}
\begin{proof}
Take arbitrary vertices $u,u' \in U$ and a vertex $u_1$ in a triangle. Since $\{u_1\}^{\sarr\sarr}$ contains that triangle, we get $\{u_1\}^{\sarr\sarr} = U$ by~\Cref{lem:tree_def_bipartite}. In particular, $u_1 \sarr u_2 \sarr u$ for some $u_2$, and then also $u_1 \sarr u_3 \sarr u_2$ for some $u_3$. Now $u_2$ is in a triangle (namely $u_1$, $u_2$, $u_3$) and we can use the same reasoning to show that $u$ is in a triangle. Applying the argument once more, we get $u_4$ such that $u \sarr u_4 \sarr u'$ -- the required common neighbor.
\end{proof}

We call a triple $(U_0, U_1, U_2)$ of subsets of $U$ a \emph{strong configuration} if
	
\begin{enumerate}[label=(\alph*)]
\item\label{it:u1} each $U_i$ is tree-definable,
\item\label{it:u2} each $U_i$ is independent (i.e., $\psarr{U_i} \cap U_i=\emptyset$), 
\item\label{it:u3} $\psarr{U_i} \supseteq U_j$ for all $i,j$ with $i \neq j$, and
\item\label{it:u5} $\OG|_{U_0\cup U_1}$ and $\OG|_{U_0\cup U_2}$ are both (weakly) connected.
\end{enumerate}

We start with a strong configuration $(U_0, U_1, U_2)$ = $(\{u_0\},\{u_1\},\{u_2\})$, where $u_0$, $u_1$, and $u_2$ form a triangle. 
The strategy now is to gradually enlarge the sets $U_i$~%
(preserving~\cref{it:u1}--\cref{it:u5}) so that, eventually,  $U_0 \cup U_1 \cup U_2 = U$. 
If this is achieved, then $(\bigcup U_0, \bigcup U_1, \bigcup U_2, \bigcup U=A)$ clearly forms a triangle configuration for $\arr$ (note that the $U_i$ are disjoint by the other conditions)  and all the four sets are pp-definable from $\arr$ and 1-orbits of $\group G$ --- a contradiction would be reached. 

We grow the sets by applying~\Cref{blowup} or~\Cref{lem:INFsecond};
more precisely we apply~\Cref{blowup} as long as it enlarges a set.
If~\Cref{blowup} fails to enlarge the configuration, 
we apply~\Cref{lem:INFsecond}.
If neither operation enlarges the configuration $\bigcup U = A$ and we obtained our goal.

We assume that
\begin{itemize}
\item $(U_0, U_1, U_2)$ is a strong configuration, 
\item some $u_i \in U_i$ form a triangle, and 
\item $U_0 \cup U_1 \cup U_2$ is a proper subset of $U$,
\end{itemize}
and present the first lemma:

\begin{lemma}\label{blowup}
	 Let $i\in \{0,1,2\}$, and let  $U_i' =  \psarr{U_{i-1}} \cap \psarr{U_{i+1}}$ and $U_j'\coloneqq U_j$ for $j \neq i$ (indices are computed modulo 3). Then $(U_0',U_1',U_2')$ is a strong configuration and $U_j \subseteq U'_j$ for each $j$.
\end{lemma}

\begin{proof}
Let $j,k$ be such that $\{i,j,k\}=\{0,1,2\}$ and $k \neq 0$. Thus $U'_i = \psarr{U_j} \cap \psarr{U_k}$, $U'_j = U_j$, $U'_k=U_k$, and $\OG_{U_i \cup U_j}$ is connected.

The sets $\psarr{U_j}$ and $\psarr{U_k}$ both contain $U_i$, therefore so does $U_i'$. 
Moreover, $\psarr{(U_i')} \supseteq \psarr{U_i} \supseteq U_j, U_k$ and $\psarr{U_j},\psarr{U_k} \supseteq \psarr{U_j} \cap \psarr{U_k} = U'_i$. Since $U_i'$ is tree-definable, conditions~\ref{it:u1}, \ref{it:u3} and the inclusions in the statement are verified.

Since $\OG_{U_i \cup U_j}$ is connected and $U_i$, $U_j$ are independent, any two vertices in $U_j$ are connected  by a walk in $U_i \cup U_j$ of even length. Observe that every vertex of $U_i' \subseteq \psarr{U_j}$ is adjacent to a vertex in $U_j$. It follows that $\OG_{U_i' \cup U_j}$ is connected but also that $U_i'$ is independent. Indeed, otherwise we obtain a walk of odd length in $\psarr{U_k} \supseteq U_i' \cup U_j$, which is a proper subset of $U$ (as $U_k$ is independent), a contradiction with~\Cref{lem:tree_def_bipartite}.

We have verified condition~\ref{it:u2} and a half of~\ref{it:u5}, the other half of~\ref{it:u5} readily follows.
\end{proof}

\noindent
If one of the inclusions in~\Cref{blowup} is proper, we succeeded in expanding our configuration. Assume, therefore, that 
$$
U_i = \psarr{U_{i-1}} \cap \psarr{U_{i+1}} \ \mbox{ for each } i \enspace.
$$
Set
$$
V_i\coloneqq \psarr{U_i} \ \setminus \ (U_{i-1}\cup U_{i+1}) \ \mbox{ for } i\in \{0,1,2\}
$$
and note that we do not claim $V_i$ is tree-(or pp-)definable.

Now the graph has the following structure. All the $U_i$ and $V_i$ are pairwise disjoint, each $u \in U_i$ has an edge to  $U_{i-1}$ and to $U_{i+1}$, and has no edge to $V_{i-1}$ or $V_{i+1}$. Every $v \in V_i$ has an edge to $U_i$ (and no edge to $U_{i-1}$ or $U_{i+1}$). The next lemma shows that $v \in V_i$ has an edge to $V_{i-1}$ and an edge to $V_{i+1}$. 
Its proof, given in~\Cref{apps:cutoff}, is exceptional in that it works with the original digraph $\relstr{A}$, unlike all the other lemmata.
This is, in a way, necessary, because a triangle in the bowtie graph provably cannot be properly expanded by means of tree definitions.

\begin{restatable}{lemma}{cutoff}\label{cutoff}
	If $i\neq j$, then $\psarr{V_j} \supseteq V_i$.
\end{restatable}

\begin{lemma} \label{lem:vi_nonvoid}
Each $V_i$ is nonempty. 
\end{lemma}
\begin{proof}
   Some $V_i$ must be nonempty, since otherwise $\psarr{U_i} = U_{i-1} \cup U_{i+1}$ for each $i$, and then $U_1^{\sarr\sarr} = U_1 \cup U_2 \cup U_3$ is a proper tree-definable subset of $U$ that contains a triangle, a contradiction to~\Cref{lem:tree_def_bipartite}.
   By~\Cref{cutoff}, each vertex in $V_i$ has a neighbor in $V_j$ for any $j \neq i$; in particular, $V_j$ is nonempty.
\end{proof}

While each $U_i$ is independent,
we now observe that 
the $V_i$ are quite different.

\begin{lemma}\label{lem:good_neighbour} 
Each common neighbor of $v \in V_i$ and $u \in U_i$ is in $V_i$. In particular, each $v \in V_i$ has a neighbor in $V_i$. 
\end{lemma}

\begin{proof}
   The common neighbor is in $\psarr{U_i} = U_{i-1} \cup U_{i+1} \cup V_i$. But there are no edges between $V_i$ and $U_{i-1} \cup U_{i+1}$. The second part follows from \Cref{lem:common_neighbors}.
\end{proof}

We have all the necessary structural information to expand our configuration.
By \Cref{lem:vi_nonvoid}, there exists a vertex  $v_1 \in V_1$. %
We fix such a vertex and inductively define
$$
W = \psarr{\{v_1\}} \cap \psarr{U_1}, \ 
S_0 := \psarr{W} \cap \psarr{U_0}, \ 
S_{n+1} := \psarr{S_n} \cap \psarr{U_0}
$$
for $n=1, 2, \dots$. We will show that $(U_0, S_n, S_{n+1})$ is a strong configuration properly extending $(U_0, U_1, U_2)$ for a sufficiently large even $n$. Observe first that each $S_n$ is tree-definable and let's move on to more interesting facts.

\begin{lemma} \label{lem:inclusions}
The following inclusions hold.
\begin{eqnarray*}
U_1 \varsubsetneq S_0 \subseteq S_2 \subseteq \dots \subseteq S_{2n} \subseteq \dots \subseteq U_1 \cup V_0 \\ 
U_2 \varsubsetneq S_1 \subseteq S_3 \subseteq \dots \subseteq S_{2n+1} \subseteq \dots  \subseteq U_2 \cup V_0
\end{eqnarray*}
\end{lemma}

\begin{proof}
  We begin by proving $U_1\varsubsetneq S_0\subseteq U_1\cup V_0$:
Each vertex in $W$ is adjacent to $v_1 \in V_1$ and a vertex in $U_1$, so it must belong to $V_1$ by~\Cref{lem:good_neighbour}. Since $\psarr{U_0} = U_1 \cup U_2 \cup V_0$ and there are no edges between $U_2$ and $V_1$, we get $S_0 \subseteq U_1 \cup V_0$.   

By~\Cref{lem:common_neighbors},  $v_1$ and each $u_1 \in U_1$ have a common neighbor, which  belongs to $W$, therefore $u_1 \in \psarr{W}$. Since $u_1 \in \psarr{U_0}$, we have shown that $U_1 \subseteq S_0$.
Moreover, the inclusion is proper as each vertex in $W \subseteq V_1$ has a neighbor in $V_0$ by~\Cref{cutoff}; this neighbor belongs to $S_0$. 

The proof is finished by induction: e.g., for even $n$, $S_{n+1} \subseteq \psarr{S_n} \cap \psarr{U_0} \subseteq \psarr{(U_1 \cup V_0)} \cap \psarr{U_0} \subseteq U_2 \cup V_0$, and every vertex in $S_n$ has a neighbor in $S_{n+1}$ (use~\Cref{lem:good_neighbour} for vertices in $V_0$), so $S_{n} \subseteq S_{n+2}$. 
\end{proof}

\begin{lemma} \label{lem:s_is_independent}
Every $S_n$ is independent.
\end{lemma}
\begin{proof}
We know that $S_n \subseteq U_i \cup V$ (where $i \in \{1,2\}$) by~\Cref{lem:inclusions}, that $U_i$ is independent, and that there are no edges between $V_0$ and $U_i$. It is therefore enough to verify that there are no edges in $S_n \cap V_0$.
  Assume to the contrary that $v,v'$ are adjacent and $\{v,v'\}\subseteq S_n\cap V_0$;
  clearly $\{v,v'\}\subseteq S_{n+1}\cap V_0$.
  On the one hand, $\psarr{(S_n\cap S_{n+1})}$ contains a triangle by~\Cref{lem:common_neighbors}.
  On the other hand, $S_n\cap S_{n+1}\subseteq V_0$ and therefore $\psarr{(S_n\cap S_{n+1})}\cap U_1 = \emptyset$,
  which contradicts~\Cref{lem:tree_def_bipartite}.
\end{proof}

\Cref{lem:inclusions} shows that the triple $(U_0, S_n, S_{n+1})$ properly extends $(U_0, U_1, U_2)$ for any even $n$. The following lemma thus finishes the proof.

\begin{lemma}\label{lem:INFsecond}
The triple $(U_0, S_n, S_{n+1})$ is a strong configuration for every sufficiently large even $n$.
\end{lemma}
\begin{proof}
  We have already observed that each $S_n$ is tree-definable, so condition~\ref{it:u1} holds.

  It follows from the inclusions in \Cref{lem:inclusions} that $S_n = S_{n+2}$ for every sufficiently large $n$. Pick such an even $n$. 
By~\Cref{lem:s_is_independent}, $S_n$ and $S_{n+1}$ are independent, proving condition~\ref{it:u2}. As for the inclusions in condition~\ref{it:u3}, we have that $\psarr{U_0}$ contains both $S_n$ and $S_{n+1}$ by definitions, that $\psarr{S_n}$ contains $U_0$ (as it contains $\psarr{U_1}$ by~\Cref{lem:inclusions}) and $S_{n+1}$ (by definition of $S_{n+1}$), and that $\psarr{S_{n+1}}$ contains $U_0$ and $S_{n+2}$, which is equal to $S_n$. 
The remaining, connectivity condition~\ref{it:u5} is also simple: $\OG_{U_0 \cup S_n}$ is connected, because every vertex of $S_n$ is adjacent to a vertex of $U_0$ (as $\psarr{U_0} \supseteq S_n$ by definition), $S_n$ contains $U_1$, and $\OG_{U_0 \cup U_1}$ is connected. For a similar reason, $\OG_{U_0 \cup S_{n+1}}$ is connected as well, and the proof is concluded.  
\end{proof}

\subsection{Weakest pseudoloop conditions}

\section{Countably categorical structures without pseudo-wnu  polymorphisms}\label{sect:no-pwnu}
In this section we construct an \oc model-complete core structure $\fa$ that does not pp-interpret $K_3$ with parameters, and whose polymorphism clone does not contain any pseudo-WNU operation of any arity. This provides a counterexample to~\cite[Problem 14.2.6 (21)]{Book}.
In fact, we show the following stronger result: let $\Sigma$ be a %
 balanced minor condition, and let  $\bigvee_{i\in\omega}\Delta_i$ %
  be a \emph{weak equational condition}, i.e., %
  a disjunction of equational conditions $\Delta_i$.  %
   We prove that if for each fixed $i\in\omega$, the satisfaction of $\Sigma$  does not imply the satisfaction of $\Delta_i$ over finite idempotent polymorphism clones, then there exists an \oc model-complete core structure $\fa$ %
such that $\pol(\fa)$ satisfies $\overline\Sigma$ while omitting every $\Delta_i$;  that is, $\pol(\fa)$ does not satisfy the disjunction of the $\Delta_i$ even if this disjunction might well be implied by $\Sigma$ over finite idempotent polymorphism clones.

Moreover, the \emph{orbit growth}, i.e., the growth of the number of $n$-orbits as $n$ increases, can be taken to be smaller than doubly exponential.
It was shown in \cite{BKOPP,BKOPP-equations} that if $\relstr A$ is an $\omega$-categorical model-complete core whose orbit growth is smaller than $2^{2^n}$, then $\relstr A$ pp-interprets $K_3$ with parameters if, and only if, there exists a finite subset of $A$ on which $\pol(\relstr A)$ does not satisfy any  non-trivial minor condition.
Thus, our structure in~\Cref{thm:three} locally admits polymorphisms satisfying non-trivial minor conditions while still avoiding pseudo-WNU polymorphisms.

We now define some basic notions that we borrow from model theory.
A structure $\relstr A$ is \emph{homogeneous} if for every finite set $B\subseteq A$ and every embedding $f\colon\relstr B\to\relstr A$ of the structure $\relstr B$ induced by $\relstr A$ on $B$,  there exists an automorphism $\alpha$ of $\relstr A$ such that $\alpha|_B=f$.
Homogeneous structures are uniquely identified by the class of their finite substructures, which is called their \emph{age}.
Moreover, a classical result by Fra\"iss\'e's states that a countable class of finite relational structures $\mathcal C$ is the age of a homogeneous structure $\relstr C$ iff $\mathcal C$ is closed under taking substructures and satisfies the so-called \emph{amalgamation property}: for all structures $\relstr B,\relstr C_1,\relstr C_2\in\mathcal C$, and all embeddings $f_i\colon \relstr B\to \relstr C_i$, there exists a structure $\relstr D\in\mathcal C$ and embeddings $e_i\colon\relstr C_i\to\relstr D$ such that $e_1\circ f_1=e_2\circ f_2$. The structure $\relstr C$ is called the \emph{Fra\"iss\'e limit} of $\mathcal C$.
In the case that the embeddings $e_1,e_2$ can always be chosen so that $e_1(C_1)\cap e_2(C_2)=e_1(f_1(B))$ holds, then we say that $\mathcal C$ has the \emph{strong amalgamation property (SAP)}.

Fix a finite relational structure $\fa$ with domain $\{1,\dots,n\}$, and let $k\geq 2$. For an arbitrary set $B$, let $[B]^k$ be the set of tuples $(b_1,\dots,b_k)\in B^k$ with pairwise distinct entries. Let $\sigma$ be a relational signature that contains a symbol $R^k$ of arity $kr$ for every relation $R$ of arity $r$ of $\fa$, together with a $2k$-ary symbol $\sim$, and $k$-ary symbols $P_1,\ldots, P_n$.
Let $\mathcal C(\fa,k)$ be the class of all finite substructures of  $\sigma$-structures $\fb$ satisfying the following conditions:
\begin{itemize}
	\item $\sim$ is an equivalence relation on $[B]^k$ with $n$ classes $P_1,\ldots,P_n$; 
	\item identifying $P_i$ with $i$, and denoting by $[x]_\sim$ the equivalence class of $x$ for every $x\in [B]^k$ we have:  
	for every relation $R$ of $\fa$, say of arity $r$, and for every $b^1,\dots,b^r\in B^k$, one has $(b^1,\dots,b^r)\in R^k$ if, and only if, $b^1,\dots,b^r\in [B]^k$ and $([b^1]_\sim,\dots,[b^r]_\sim)\in R$. %
\end{itemize}
Note that any structure in $\mathcal C(\fa,k)$ is uniquely determined by $P_1,\ldots,P_n$.  
It can be seen that $\mathcal C(\fa,k)$ is nonempty and has the SAP, and therefore its Fra\"iss\'e limit $\fa^{\otimes k}$ is a homogeneous structure without algebraicity.
Observe that the factor map corresponding to $\sim$ is a pp-interpretation of $\fa$ in $\fa^{\otimes k}$ (again identifying $P_i$ with $i$, as we shall often do in the following); here, we use that $\neq$ is pp-definable in $\fa^{\otimes k}$ as a projection of $\sim$, and therefore the domain of this factor map is pp-definable in $\fa^{\otimes k}$.
 Moreover, if two $m$-tuples $a,b$ from $\fa^{\otimes k}$ are such that they satisfy the same equalities among their components and are such that $[a']_\sim=[b']_\sim$ for all projections $a',b'$ of $a,b$ onto the same $k$ coordinates which are injective, then $a$ and $b$ belong to the same orbit under the action of $\aut(\fa^{\otimes k})$ on $k$-tuples.
In particular, $\fa^{\otimes k}$ is $\omega$-categorical. The proof of the following is deferred to~\Cref{app:pwnu}.%

\begin{restatable}{proposition}{corecover}\label{prop:mc_core}
If $\fa$ is a core, then $\fa^{\otimes k}$ is a model-complete core.
\end{restatable}

\begin{proposition}\label{prop:structontuples}
Let $\Sigma$ be a balanced minor condition that is satisfiable in $\pol(\fa)$ by idempotent operations.
Then $\pol(\fa^{\otimes k})$ contains injective functions satisfying $\overline\Sigma$.
\end{proposition}
\begin{proof}
As before, let $\{1,\dots,n\}$ be the domain of $\fa$, let $B$ be the domain of $\fa^{\otimes k}$, and identify each equivalence class $P_i$ of $\sim$ with $i$, for all $i\in \{1,\ldots,n\}$.
For every symbol $s$ appearing in $\Sigma$, we set $C_s:=B^{r}$, where $r$ is the arity of $s$. %
We also use $s$ to denote an idempotent operation in $\pol(\fa)$ witnessing the fact that $\pol(\fa)$ satisfies $\Sigma$. Our goal is to assign a value in $B$ to each element of $C_s$, thus obtaining a function $s'\colon B^r\to B$; the functions thus obtained will, together with embeddings for the new unary symbols, witness the satisfaction of $\overline \Sigma$.

In order to do that, we first define a partial mapping $q_s\colon [C_s]^k\to\{1,\dots,n\}$ by setting, for any pairwise distinct tuples $c^1,\ldots,c^k\in C_s=B^r$ with the property that the tuples $(c^1_i,\dots,c^k_i)$ are injective for all $i\in\{1,\ldots,r\}$ $$q_s(c^1,\dots,c^k):=s([(c^1_1,\dots,c^k_1)]_\sim,\dots,[(c^1_r,\dots,c^k_r)]_\sim).$$
We then extend $q_s$ to a total function on $[C_s]^k$ by setting its value to be $1$ elsewhere.  
Identifying the classes of its kernel with elements of the set $\{1,\ldots,n\}$, we see that $q_s$ induces a structure $\fc_s$ on $C_s$ whose finite substructures belong to $\mathcal C(\fa,k)$.
Since $\fa^{\otimes k}$ is $\omega$-categorical, we obtain that there exists an embedding $s'\colon\fc_s\to\fa^{\otimes k}$.  
 Since $s$ is an idempotent polymorphism of $\fa$, the identity map on $B^r=C_s$ is a homomorphism from  $(\fa^{\otimes k})^r$ to $\fc_s$. Hence, $s'$ is, viewed as the composition of an embedding with that identity map, a polymorphism of $\fa^{\otimes k}$; moreover, it is injective.

Let $s,t$ be symbols of arities $r_s, r_t$ which appear in $\Sigma$, and let  $\sigma\colon[r_s]\to[r]$ and $\tau\colon[r_t]\to [r]$, where $r\geq 1$.
We prove that if $s^\sigma\approx t^\tau$ is an identity in $\Sigma$, then $u\circ s'^\sigma = v\circ t'^\tau$ holds for some embeddings $u,v$ of $\fa^{\otimes k}$.
Let $m\geq k$, and let $b_1,\dots,b_r$ be $m$-tuples of elements of $\fa^{\otimes k}$.
Observe that the $m$-tuples $(s')^\sigma(b_1,\dots,b_r)$ and $(t')^\tau(b_1,\dots,b_r)$ satisfy the same equalities since $\Sigma$ is balanced and since both $s'$ and $t'$ are injective.
Let $i_1,\dots,i_k\in\{1,\ldots,m\}$ be distinct, and let $c_j$ be the $k$-tuple obtained by projecting $b_j$ onto $i_1,\dots,i_k$, for all $j\in\{1,\ldots,r\}$.
We claim that if $(s')^\sigma(c_1,\dots,c_r)$ is injective, then $(s')^\sigma(c_1,\dots,c_r)$ and $(t')^\tau(c_1,\dots,c_r)$ belong to the same $\sim$-class. If that is the case, then  $(s')^\sigma(b_1,\dots,b_r)$ and $(t')^\tau(b_1,\dots,b_r)$ are in the same orbit under $\aut(\fa^{\otimes k})$, by the definition of $\fa^{\otimes k}$ and its homogeneity. With this, and since $m\geq k$ was arbitrary, a standard compactness argument (see e.g.~\cite[Lemma~3]{canonical}) yields the existence of embeddings $u,v$ of $\fa^{\otimes k}$ satisfying $u\circ (s')^\sigma = v\circ (t')^\tau$.
This proves that $\overline\Sigma$ is satisfiable in $\pol(\fa^{\otimes k})$.

For the claim that remains to be proven, let $q_s,q_t$ be as in the definition of $s',t'$ from $s,t$. Note that by definition of $q_s$  we have that $q_s(c_{\sigma(1)},\ldots,c_{\sigma(r_s)})$ equals $1$ whenever at least one of the tuples $c_1,\dots,c_r$ is not injective; the analogous statement is true for $t$ and $\tau$. If, on the other hand, all of  $c_1,\dots,c_r$ are injective, then $q_s(c_{\sigma(1)},\ldots,c_{\sigma(r_s)})$ is equal to  $s([(c^1_{\sigma(1)},\dots,c^k_{\sigma(1)})]_\sim,\dots,[(c^1_{\sigma(r_s)},\dots,c^k_{\sigma(r_s)})]_\sim)$, which is $s^\sigma([(c^1_{1},\dots,c^k_{1})]_\sim,\dots,[(c^1_{r},\dots,c^k_{r})]_\sim)$. Since this  equals  $t^\tau([(c^1_{1},\dots,c^k_{1})]_\sim,\dots,[(c^1_{r},\dots,c^k_{r})]_\sim)$ by the identities in $\Sigma$, going back the same argument with $t$ and $\tau$ we see that $q_s(c_{\sigma(1)},\ldots,c_{\sigma(r_s)})$ and $q_t(c_{\tau(1)},\ldots,c_{\tau(r_t)})$ are equal, as in the non-injective case. But $s', t'$ preserve $\sim$-classes, and so it follows that the $\sim$-classes of $s'(c_{\sigma(1)},\ldots,c_{\sigma(r_s)})$ and $t'(c_{\tau(1)},\ldots,c_{\tau(r_t)})$ agree. Since  $(s')^\sigma(c_1,\dots,c_r)=s'(c_{\sigma(1)},\ldots,c_{\sigma(r_s)})$, and since the analogous statement holds for $t$ and $\tau$, our claim follows.
\end{proof}

Let $(\relstr A_i)_{i\in\omega}$ be a sequence of structures, and $(k_i)_{i\in\omega}$ be a sequence of positive natural numbers.
Let $\sigma_i$ be the signature used in the construction of $\mathcal C(\relstr A_i,k_i)$, and assume without loss of generality that the signatures $(\sigma_i)_{i\in\omega}$ are all disjoint.
The \emph{superposition} of the classes $\mathcal C(\relstr A_i,k_i)$ is the class of structures $\fb$ in the signature $\bigcup\sigma_i$ whose $\sigma_i$-reduct belongs to $\mathcal C(\relstr A_i,k_i)$ for all $i\in\omega$.
It is a standard fact that the superposition of classes that have SAP has itself SAP.
A straightforward modification of the proof of~
\Cref{prop:structontuples} shows that the Fra\"iss\'e limit of the superposition of classes $\mathcal C(\relstr A_i,k_i)$ also admits injective polymorphisms satisfying $\overline{\Sigma}$.
The proof can be found in~\Cref{app:pwnu} for the convenience of the reader.
\begin{restatable}{proposition}{superposition}\label{prop:superposition}
	Let $(\fa_i)_{i\in\omega}$ be a sequence of structures, each having idempotent polymorphisms satisfying $\Sigma$.
	Let $(k_i)_{i\in\omega}$ be a sequence of positive natural numbers.
	The Fra\"iss\'e limit of the superposition of all the classes $\mathcal C(\fa_i,k_i)$ has injective polymorphisms satisfying $\overline\Sigma$.
\end{restatable}

\begin{theorem}\label{thm:counterexample_full}
	Let $\Sigma$ be a %
	 balanced minor condition, and let $\bigvee_{i\in\omega} \Delta_i$ be a weak equational condition such that for every $i\in\omega$, there exists a finite idempotent polymorphism clone satisfying $\Sigma$ and not satisfying $\Delta_i$. Then 
	there exists an \oc homogeneous model-complete core structure $\fa$ with small orbit growth such that $\pol(\fa)$ satisfies $\overline\Sigma$ and does not satisfy  $\Delta_i$ for any $i\in\omega$.
\end{theorem}
\begin{proof}
	Let $(\fa_i)_{i\in\omega}$ be a sequence of finite structures such that $\pol(\fa_i)$ is a finite idempotent clone that satisfies $\Sigma$ and that does not satisfy $\Delta_i$.
	Let $k\colon\mathbb N\to\mathbb N$ be a function increasing sufficiently fast.
	Let $\fa$ be the Fra\"iss\'e limit of the superposition of all the classes $\mathcal C(\fa_i,k(i))$. 
	Then $\fa$ is $\omega$-categorical and even has small orbit growth (see~\cite{TopologyIsRelevant}, or Lemma 5.6 in~\cite{TopoRelevant}). By definition, it is homogeneous, and as in~\Cref{prop:mc_core}, one sees that it is a model-complete core since the Fra\"{i}ss\'{e} limit of each of the superposed classes is. By~\Cref{prop:structontuples} and \Cref{prop:superposition}, $\fa$ satisfies $\overline{\Sigma}$.
	
	Now for any $i\in\omega$, since $\pol(\fa)$ pp-interprets $\fa_i$, we can refer to~\cite{Topo-Birk} to conclude that %
	 $\pol(\fa)$ does not satisfy $\Delta_i$.
\end{proof}

\begin{proof}[Proof of~\Cref{thm:three}]
For any fixed $n\geq 3$, the set of identities stipulating the existence of an $n$-ary pseudo-WNU function is not implied by the minor condition of containing a Siggers function over finite idempotent polymorphism clones. The structure $\relstr A$ obtained by applying~\Cref{thm:counterexample_full} cannot pp-interpret  $K_3$ with parameters, since the pseudo-Siggers identity satisfied by its polymorphisms prevents this~\cite{Topo}; here we use the fact that it is a model-complete core.

The example can be made to have a finite signature by using the Hrushovski-encoding from~\cite{HrushovskiEncoding, GJKMP-conf}: the encoding is a structure $\fb$ with a finite relational signature such that $\pol(\fb)$ satisfies every pseudo-variant of any minor condition that is satisfied in $\pol(\fa)$ by injections~\cite[Proposition~3.16]{HrushovskiEncoding}, in particular the pseudo-Siggers condition.
On the other hand, there exists a pp-interpretation of $\fa$ in   $\fb$ by~\cite[Proposition 3.13]{HrushovskiEncoding} combined with~\cite{Topo-Birk}, 
which yields the absence of pseudo-WNU operations in $\pol(\fb)$. 
The structure $\fb$  still has slow orbit growth~\cite[Proposition 3.12]{HrushovskiEncoding}. Since the structure $\fa$ is a homogeneous model-complete core, its encoding $\fb$ can be made so that it is a model-complete core as well: being homogeneous itself $\fa$ need not be homogenized for the encoding, and the new relations used in the  encoding can be enriched by relations for their complement (the proof of SAP in~\cite[Lemma 3.5]{HrushovskiEncoding} still works). Then all endomorphisms of $\fb$ are embeddings, and the \emph{decoding blow up} of $\fb$   is a homogeneous expansion by pp-definable relations; thus, $\fb$ is a model-complete core. Hence, as was the case before the encoding, it cannot pp-interpret $K_3$ with parameters by~\cite{Topo}. This situation would not change for an expansion of $\fb$ by orbits, since $\fb$ is a model-complete core and all orbits are pp-definable anyway. The signature of $\fb$ can further be reduced to a single relation by replacing its relations, say $R_1,\dots,R_k$, by %
$R_1\times\dots\times R_k$, giving the hypergraph from~\Cref{thm:three}.
\end{proof}

\section{Conclusion} \label{sec:conclusion}

We conclude with %
open problems related to our goal to further develop a structural theory amenable to infinite structures.

\subsubsection{Directed graphs} \Cref{thm:two} and \Cref{thm:one_a} apply to  digraphs, but only finite ones. On the other hand, \Cref{thm:one_b} applies to  infinite graphs, but the factor graph has to be symmetric. A tantalizing  direction is to generalize the latter theorem to non-symmetric factor digraphs of algebraic length 1. Our results suggest three approaches toward this goal: to improve our novel relational approach to the finite theorems, to further exploit our infinite-to-finite reduction technique applied in \Cref{thm:triangle_config}, and, finally, to develop from our proofs an alternative, algebraic approach. The last outcome would be the most desired one, since the algebraic techniques, so powerful in the finite (and substantially influenced by the non-symmetric generalization~\cite{Cyclic} of \cite{HellNesetril}), remain relatively weak in the infinite.

\subsubsection{A uniform identity}

\Cref{thm:ugly_terms} gives us, for every finite domain, a single identity  satisfied in %
 the polymorphism clone of any core structure on this domain which does not pp-interpret all finite structures.  It would be desirable to obtain a single identity for all finite domains, independently of the size. This could pave the way for a positive answer to the open problem stated in~\cite{Topo} and~\cite{BPP-projective-homomorphisms} whether  the failure of any non-trivial algebraic invariant for $\fa$ leads to  the existence  of a pp-interpretation of $K_3$~\cite{Topo-Birk}. In fact, this could  yield a single algebraic  invariant for $\omega$-categorical model-complete cores  witnessing the failure of the pp-interpretation (without parameters!) of EVERYTHING.

\subsubsection{Pseudo-WNUs} While the negative result in~\Cref{thm:three} might disappoint hopes 
sparked by the result of Barto and Pinsker~\cite{Topo} for an algebraic theory of polymorphism clones of $\omega$-categorical structures, not all is lost for the smaller subclass of first-order reducts of finitely bounded homogeneous structures for which a CSP complexity dichotomy has been conjectured. All complexity classifications within that class have shown the existence of pseudo-WNUs in the tractable cases (see e.g.~\cite{Book, infinitesheep} for a recent account of results), and recently more general algebraic methods have been developed for these classifications~\cite{SmoothApproximations}. One of the main open problems is whether this situation generalizes to the entire class.

\bibliographystyle{IEEEtran}
\bibliography{local}

\def\cprime{$'$} \def\cprime{$'$} \def\cprime{$'$}
\begin{thebibliography}{10}
\providecommand{\url}[1]{#1}
\csname url@samestyle\endcsname
\providecommand{\newblock}{\relax}
\providecommand{\bibinfo}[2]{#2}
\providecommand{\BIBentrySTDinterwordspacing}{\spaceskip=0pt\relax}
\providecommand{\BIBentryALTinterwordstretchfactor}{4}
\providecommand{\BIBentryALTinterwordspacing}{\spaceskip=\fontdimen2\font plus
\BIBentryALTinterwordstretchfactor\fontdimen3\font minus
  \fontdimen4\font\relax}
\providecommand{\BIBforeignlanguage}[2]{{%
\expandafter\ifx\csname l@#1\endcsname\relax
\typeout{** WARNING: IEEEtran.bst: No hyphenation pattern has been}%
\typeout{** loaded for the language `#1'. Using the pattern for}%
\typeout{** the default language instead.}%
\else
\language=\csname l@#1\endcsname
\fi
#2}}
\providecommand{\BIBdecl}{\relax}
\BIBdecl

\bibitem{HellNesetril}
P.~Hell and J.~Ne\v{s}et\v{r}il, ``On the complexity of {H}-coloring,''
  \emph{Journal of Combinatorial Theory, Series B}, vol.~48, pp. 92--110, 1990.

\bibitem{BartoKozikNiven}
L.~Barto, M.~Kozik, and T.~Niven, ``The {CSP} dichotomy holds for digraphs with
  no sources and no sinks (a positive answer to a conjecture of {B}ang-{J}ensen
  and {H}ell),'' \emph{SIAM Journal on Computing}, vol.~38, no.~5, 2009.

\bibitem{Bulatov}
A.~A. Bulatov, ``A dichotomy theorem for constraint satisfaction problems on a
  3-element set,'' \emph{Journal of the ACM}, vol.~53, no.~1, pp. 66--120,
  2006.

\bibitem{ZhukFVConjecture}
D.~Zhuk, ``A proof of {CSP} dichotomy conjecture,'' in \emph{58th {IEEE} Annual
  Symposium on Foundations of Computer Science, {FOCS} 2017, Berkeley, CA, USA,
  October 15-17, 2017}, 2017, pp. 331--342, https://arxiv.org/abs/1704.01914.

\bibitem{wonderland}
L.~Barto, J.~Opr\v{s}al, and M.~Pinsker, ``The wonderland of reflections,''
  \emph{Israel Journal of Mathematics}, vol. 223, no.~1, pp. 363--398, 2018.

\bibitem{Siggers}
M.~H. Siggers, ``A strong {M}al'cev condition for varieties omitting the unary
  type,'' \emph{Algebra Universalis}, vol.~64, no.~1, pp. 15--20, 2010.

\bibitem{Book}
M.~Bodirsky, ``Complexity of infinite-domain constraint satisfaction,'' 2020,
  submitted for publication in the LNL Series, Cambridge University Press.

\bibitem{infinitesheep}
\BIBentryALTinterwordspacing
M.~Pinsker, ``Current challenges in infinite-domain constraint satisfaction:
  Dilemmas of the infinite sheep,'' in \emph{2022 IEEE 52nd International
  Symposium on Multiple-Valued Logic (ISMVL)}.\hskip 1em plus 0.5em minus
  0.4em\relax Los Alamitos, CA, USA: IEEE Computer Society, 2022, pp. 80--87.
  [Online]. Available:
  \url{https://doi.ieeecomputersociety.org/10.1109/ISMVL52857.2022.00019}
\BIBentrySTDinterwordspacing

\bibitem{Topo}
L.~Barto and M.~Pinsker, ``Topology is irrelevant,'' \emph{SIAM Journal on
  Computing}, to appear. Preprint arXiv:1602.04353v3; an extended abstract
  appeared in the proceedings of LICS'16 under the title `The algebraic
  dichotomy conjecture for infinite domain constraint satisfaction problems'.

\bibitem{BulatovHColoring}
A.~A. Bulatov, ``{H}-coloring dichotomy revisited,'' \emph{Theoretical Computer
  Science}, vol. 349, no.~1, pp. 31--39, 2005.

\bibitem{BPP-projective-homomorphisms}
M.~Bodirsky, M.~Pinsker, and A.~Pongr\'acz, ``Projective clone homomorphisms,''
  2014, accepted for publication in the Journal of Symbolic Logic, Preprint
  arXiv:1409.4601.

\bibitem{Pseudo-loop}
P.~Gillibert, J.~Jonu\v{s}as, and M.~Pinsker, ``Pseudo-loop conditions,''
  \emph{Bulletin of the London Mathematical Society}, vol.~51, no.~5, pp.
  917--936, 2019.

\bibitem{MarotiMcKenzie}
M.~Mar\'oti and R.~McKenzie, ``Existence theorems for weakly symmetric
  operations,'' \emph{Algebra Universalis}, vol.~59, no.~3, 2008.

\bibitem{Cyclic}
L.~Barto and M.~Kozik, ``Absorbing subalgebras, cyclic terms and the constraint
  satisfaction problem,'' \emph{Logical Methods in Computer Science}, vol. 8/1,
  no.~07, pp. 1--26, 2012.

\bibitem{Strong-Zhuk}
\BIBentryALTinterwordspacing
D.~Zhuk, ``Strong subalgebras and the constraint satisfaction problem,''
  \emph{CoRR}, vol. abs/2005.00593, 2020. [Online]. Available:
  \url{https://arxiv.org/abs/2005.00593}
\BIBentrySTDinterwordspacing

\bibitem{SmoothApproximations}
\BIBentryALTinterwordspacing
A.~Mottet and M.~Pinsker, ``Smooth approximations and {CSP}s over finitely
  bounded homogeneous structures,'' in \emph{Proceedings of {LICS'22}}, 2022,
  to appear. [Online]. Available: \url{https://arxiv.org/abs/2011.03978}
\BIBentrySTDinterwordspacing

\bibitem{Cores-journal}
M.~Bodirsky, ``Cores of countably categorical structures,'' \emph{Logical
  Methods in Computer Science ({LMCS})}, vol.~3, no.~1, pp. 1--16, 2007.

\bibitem{BoKaKoRo}
V.~G. Bodnar\v{c}uk, L.~A. Kalu\v{z}nin, V.~N. Kotov, and B.~A. Romov, ``Galois
  theory for {Post} algebras, part {I} and {II},'' \emph{Cybernetics}, vol.~5,
  pp. 243--539, 1969.

\bibitem{CSPSurvey}
A.~A. Bulatov, P.~Jeavons, and A.~A. Krokhin, ``The complexity of constraint
  satisfaction: An algebraic approach (a survey paper),'' \emph{In: Structural
  Theory of Automata, Semigroups and Universal Algebra (Montreal, 2003), NATO
  Science Series II: Mathematics, Physics, Chemistry}, vol. 207, pp. 181--213,
  2005.

\bibitem{BodirskyNesetrilJLC}
M.~Bodirsky and J.~Ne\v{s}et\v{r}il, ``Constraint satisfaction with countable
  homogeneous templates,'' \emph{Journal of Logic and Computation}, vol.~16,
  no.~3, pp. 359--373, 2006.

\bibitem{TopologyIsRelevant}
M.~Bodirsky, A.~Mottet, M.~Ol\v{s}\'ak, J.~Opr\v{s}al, M.~Pinsker, and
  R.~Willard, ``Topology is relevant (in the infinite-domain dichotomy
  conjecture for constraint satisfaction problems),'' in \emph{Proceedings of
  the Symposium on Logic in Computer Science -- LICS'19}, 2019.

\bibitem{Ros70}
I.~G. Rosenberg, ``\"{U}ber die funktionale {V}ollst\"{a}ndigkeit in den
  mehrwertigen {L}ogiken,'' \emph{Rozpravy Ceskoslovensk\'{e} Akad. v\v{e}d,
  Ser. Math. Nat. Sci.}, vol.~80, pp. 3--93, 1970.

\bibitem{Pin02}
M.~Pinsker, ``Rosenberg's characterization of maximal clones,'' Diploma thesis,
  Technische Universit\"{a}t Wien, 2002, arXiv:0211420.

\bibitem{Qua71}
R.~W. Quackenbush, ``A new proof of {R}osenberg’s primal algebra
  characterization theorem,'' \emph{Colloquia Mathematica Societatis J\'{a}nos
  Bolyai}, vol.~28, p. 603–634, 1971.

\bibitem{absorption}
L.~Barto and M.~Kozik, ``Absorption in universal algebra and {CSP},'' in
  \emph{The Constraint Satisfaction Problem: Complexity and Approximability},
  ser. Dagstuhl Follow-Ups, vol.~7, 2017, pp. 45--77.

\bibitem{MinimalTaylor}
L.~Barto, Z.~Brady, A.~Bulatov, M.~Kozik, and D.~Zhuk, ``Minimal {T}aylor
  algebras as a common framework for the three algebraic approaches to the
  {CSP},'' in \emph{36th Annual {ACM/IEEE} Symposium on Logic in Computer
  Science, {LICS} 2021, Rome, Italy, June 29 - July 2, 2021}.\hskip 1em plus
  0.5em minus 0.4em\relax {IEEE}, 2021, pp. 1--13.

\bibitem{BKOPP}
L.~Barto, M.~Kompatscher, M.~Ol\v{s}\'{a}k, T.~V. Pham, and M.~Pinsker, ``The
  equivalence of two dichotomy conjectures for infinite domain constraint
  satisfaction problems,'' in \emph{Proceedings of the 32nd Annual {ACM/IEEE}
  Symposium on Logic in Computer Science -- LICS'17}, 2017, preprint
  arXiv:1612.07551.

\bibitem{BKOPP-equations}
------, ``Equations in oligomorphic clones and the constraint satisfaction
  problem for $\omega$-categorical structures,'' \emph{Journal of Mathematical
  Logic}, vol.~19, no.~2, p. \#1950010, 2019, an extended abstract appeared at
  the {Proceedings of the 32nd Annual {ACM/IEEE} Symposium on Logic in Computer
  Science -- LICS'17}.

\bibitem{canonical}
M.~Bodirsky and M.~Pinsker, ``Canonical functions: a proof via topological
  dynamics,'' \emph{Homogeneous Structures, A Workshop in Honour of Norbert
  Sauer's 70th Birthday, Contributions to Discrete Mathematics}, vol.~16,
  no.~2, pp. 36--45, 2021.

\bibitem{TopoRelevant}
M.~Bodirsky, A.~Mottet, M.~Ol{\v s}{\'a}k, J.~Opr{\v s}al, M.~Pinsker, and
  R.~Willard, ``Omega-categorical structures avoiding height 1 identities,''
  \emph{Trans. Amer. Math. Soc.}, vol. 374, pp. 327--350, 2021.

\bibitem{Topo-Birk}
M.~Bodirsky and M.~Pinsker, ``Topological {B}irkhoff,'' \emph{Transactions of
  the American Mathematical Society}, vol. 367, pp. 2527--2549, 2015.

\bibitem{HrushovskiEncoding}
\BIBentryALTinterwordspacing
P.~Gillibert, J.~Jonusas, M.~Kompatscher, A.~Mottet, and M.~Pinsker, ``When
  symmetries are not enough: {A} hierarchy of hard constraint satisfaction
  problems,'' \emph{{SIAM} J. Comput.}, vol.~51, no.~2, pp. 175--213, 2022.
  [Online]. Available: \url{https://doi.org/10.1137/20m1383471}
\BIBentrySTDinterwordspacing

\bibitem{GJKMP-conf}
P.~Gillibert, J.~Jonu\v{s}as, M.~Kompatscher, A.~Mottet, and M.~Pinsker,
  ``Hrushovski's encoding and $\omega$-categorical {CSP} monsters,'' in
  \emph{47th International Colloquium on Automata, Languages, and Programming,
  {ICALP} 2020, July 8-11, 2020, Saarbr{\"{u}}cken, Germany (Virtual
  Conference)}, ser. LIPIcs, vol. 168.\hskip 1em plus 0.5em minus 0.4em\relax
  Schloss Dagstuhl - Leibniz-Zentrum f{\"{u}}r Informatik, 2020, pp.
  131:1--131:17.

\end{thebibliography}

\clearpage
\appendices
\crefalias{section}{appendix}
\crefalias{subsection}{appendix}
\section{Small auxiliary results} \label{app:aux}

\subsection{Proof of~\Cref{prop:OR} from~\Cref{sec:prelim}}

\ORprop*
\begin{proof}
  Observe that $\alpha$ and $B$ are both pp-definable from $\relstr{A}$ by $\OR(\alpha,\alpha)(x,y,x,y)$ and $\alpha(x,x)$, respectively. Hence, the factor map on $B$ yields a (1-dimensional) pp-interpretation of the relation $\OR(=_C,=_C)$ on the domain $C = A/\alpha$ in $\relstr{A}$. Since $\relstr{A}$ is a core, the same map even pp-interprets a  
  structure $\relstr{C}$ on the domain $C$ containing $\OR(=_C,=_C)$ such that $\relstr{C}$ is a core. Note that as $\alpha$ is proper, $C$ has at least two elements. 

  We claim that every polymorphism of $\relstr{C}$ depends on exactly one coordinate. First note that every polymorphism depends on at least one coordinate as $\relstr{C}$ is a core. The fact that any function preserving the relation  $\OR(=_C,=_C)$ depends on at most variable is well known, and shown e.g.~in~\cite[Lemma 6.1.17]{Book}; we give the argument for the convenience of the reader. Let $n\geq 2$, and suppose that  $f$ is a function of arity $n$ preserving $\OR(=_C,=_C)$ that depends on two distinct coordinates, say without loss of generality the first two. Then there exist elements $a_1',a_1,\ldots,a_n\in C$ and $b_2',b_1,\ldots,b_n\in C$ such that $a' = f(a_1',a_2,\dotsc,a_n)\neq f(a_1,\dotsc,a_n) = a$ and such that 
    and $b = f(b_1,\dotsc,b_n)\neq f(b_1,b_2',b_3,\dotsc,b_n) = b'$. Clearly $(a_i,a'_i,b_i,b'_i)\in\OR(=_C,=_C)$ for every $i\in\{1,\ldots,n\}$, while $(a,a',b,b')\notin\OR(=_C,=_C)$, a contradiction.

If $\relstr{D}$ is any finite structure, then the map which sends every polymorphism $f$ of $\relstr{C}$ to the projection on $D$ of the same arity and onto the unique coordinate on which $f$ depends is a clone homomorphism. Hence, $\relstr{C}$ pp-interprets $\relstr{D}$, and whence so does $\relstr{A}$.
  \end{proof}

\subsection{Proof of \Cref{lem:pp_from_rpp} from \Cref{sec:finite}}

\PPfromRPP*

We instead prove a more general version.

 \begin{lemma}
Let $\relstr{A}=(A;\arr)$ be a digraph, $g \in \aut(\relstr{A})$, $\mathcal{R}$ a set of relations on $A$, and $S$ a relation on $A$. Let $\arr' = \arr + g$. If $\arr'_{01}$ and 0-ranked relations from $\mathcal{R}$ rpp-define 0-ranked $S$, then $\arr$ and $\{g^i(R): R \in \mathcal{R}, i \in \Int\}$ pp-define $S$ . 
\end{lemma}

\begin{proof}
Let $\phi'(x_1, \dots, x_n)$ be a 0-ranked rpp-formula over $\arr'_{01}$ and 0-ranked relations from $\mathcal{R}$ that defines $S$. Let $r$ be the ranking of $\phi'$. We change  each  conjunct $R'(x_1, \dots, x_k)$ with $R' \in \mathcal{R}$  to $R(x_1, \dots, x_k)$, where $R=g^{-r(x_1)}(R')$ (note that $r(x_1)=\cdots=r(x_k)$ since the relations in $\mathcal{R}$ are 0-ranked), and claim that this new formula $\phi$ is a pp-definition  of $S$ in $\relstr{A}$.

Assume that $\phi'(a_1, \dots, a_n)$. Consider a witnessing evaluation of all variables in $\phi'$ and apply to it the ranked mapping $(g^{-i})_{i \in \Int}$. The obtained evaluation of variables in $\phi$ satisfies (in $\relstr{A}$) each  conjunct involving $\arr$, since $(g^{-i})_{i\in\Int}$ is a ranked homomorphism from $(A;\arr'_{01})$ to $(A;\arr_{01})$, and each other conjunct by the definition of $\phi$. Therefore, it is a valid evaluation of $\phi'$ and we conclude that $\phi(g^{-r(x_1)}(a_1), \dots, g^{-r(x_n)}(a_n))$. But $r(x_1)= \dots = r(x_n)=0$, so $\phi(a_1, a_2, \dots, a_n)$. 

The converse implication $\phi(a_1, \dots, a_n) \Rightarrow \phi'(a_1, \dots, a_n)$ is proved analogically, using the ranked homomorphism $(g^i)_{i\in\Int}$ from $(A;\arr_{01})$ to $(A;\arr'_{01})$. 

\end{proof}

\section{Missing proofs from \Cref{sec:finite}}

\subsection{Identities for finite nonidempotent clones} \label{apps:uglyterms}

\ThmUglyTerms*

\begin{proof}
  Since the identity is nontrivial (it is not satisfied by the idempotent polymorphisms of $K_3$ -- the projections), the upward implication is trivial.
  
  For the other direction, we assume that $\relstr{A}$ does not pp-interpret some finite structure.
  Let $M$ be the subset of $A^{A^3}$ consisting of all the ternary elements of $\pol(\relstr A)$;
  note that the three projections $\pi^3_1,\pi^3_2,\pi^3_3$ belong to $M$ and that
  $\pol(\relstr A)$ acts naturally (coordinatewise) on $M$.
  The unary part of the clone is a group~%
  (since $\relstr A$ is a core)
  and, as a part of $\pol(\relstr A)$, acts coordinatewise on $M$ as well.
  Let $\group G$ be the image of this action in the group of transformations of $M$.

  In the next step, we define a digraph $\arr$ on $M$:
  we put $\pi^3_i\arr \pi^3_j$ for all $i\neq j$,
  then add $\alpha(\pi^3_1)\arr\pi^3_2$ for every unary $\alpha$ in $\pol(\relstr A)$ and finally close the relation $\arr$ under the action of $\pol(\relstr A)$ on the pairs of elements from $M$.

  We plan to apply~\Cref{thm:two_refined} to $(M,\arr)$ and $\group G$ but need to verify the assumptions first.
  The digraph $(M,\arr)$ is smooth as it includes the graph obtained by closing the triangle on the three projections under $\pol(\relstr A)$ and this one is smooth by construction. 
  Clearly, the group $\group G$ consists of automorphisms of $(M,\arr)$.
  Finally, to see that the digraph is linked we start with an arbitrary $f\in M$.
  Note that $f=f(\pi^3_1,\pi^3_2,\pi^3_3)$
  and we have
  \begin{equation*}
     f(\pi^3_1,\pi^3_2,\pi^3_3) \arr
     f(\pi^3_2,\pi^3_3,\pi^3_2) \larr
     f(\pi^3_1,\pi^3_1,\pi^3_1).
  \end{equation*}
  The last element arrows $\pi^3_2$~%
  (this edge is one of the generators as $f(\pi^3_1,\pi^3_1,\pi^3_1) = \alpha(\pi^3_1)$
  for appropriate $\alpha\in\group G$)
  and finally $\pi^3_1\arr\pi^3_2$.
  Thus every element is linked to $\pi^3_1$ and the whole digraph is linked; assumptions of~\Cref{thm:two_refined} hold.

  \Cref{thm:two_refined} cannot be satisfied with $\OR(\alpha,\alpha)$.
  Indeed in such case $(M,\arr)$ together with orbits of $\group G$ would pp-interpret, by~\Cref{prop:OR},
  every finite structure.
  As $\relstr A$ is a core the orbits of $\aut(\relstr A)$ are pp-definable in $\relstr A$ and thus orbits of $\group G$ are pp-definable as well.
  Putting things together we would be able to pp-interpret every finite structure in $\relstr A$ --- this directly contradicts our assumption.

  Thus~\Cref{thm:two_refined} holds due to a loop in $(M,\arr)$, and the loop is obtained by applying an operation from $\pol(\relstr A)$ to the generators of $(M,\arr)$. 
  This very operation~%
  (modulo standard, cosmetic changes)
  satisfies the identity in the statement of our theorem and the proof is finished.
\end{proof}

\subsection{Pseudoloops with pseudo assumptions} \label{apps:one_a_refined}

\OneARefined*

    \begin{proof}
      Let $\relstr{A}$ and $\group G$ be as in the statement, and let $\orbiteq$ be the $1$-orbit equivalence of $\group G$. We assume that $\relstr{A}/\group{G}$ has no loop and aim to prove the second item in the statement.
      
      Take a closed walk in $\relstr{A}/\group G$ of algebraic length 1:
      $$
      \quotient{a_0} \ \epsilon_0 \ \quotient{a_1} \ \epsilon_1 \ \dots \ \epsilon_{n-1} 
 \ \quotient{a_n}, \ \ \ \quotient{a_0} =\quotient{a_n} \enspace.
      $$
      Since $\arr$ is invariant under $\group G$, the representatives $a_i$ can be chosen so that in $\relstr{A}$ we have the walk
      $$
      a_0 \ \epsilon_0 \ a_1 \ \epsilon_1 \ \dots \ \epsilon_{n-1} \ a_n, \ \ \ a_0 \sim a_n \enspace,
      $$
      which is not necessarily closed. Pick $g \in \group G$ such that $g(a_n)=a_0$ and define $\arr' = \arr + g$ and $\relstr{A}'=(A;\arr')$. Let $D$ ($D'$, respectively) be the weak component of $\relstr{A}$ ($\relstr{A}'$, respectively)  containing $a_0$. We show that $\relstr{A}'|_{D'}$ is finite and has algebraic length 1.

      To see that $D'$ is finite, notice that $D$ is finite by the assumptions and that $a_n$ is in the same weak component of $\relstr{A}$ since $g(a_n)=a_0$. It follows that $D$ is invariant under $g$ which implies, by definition of $\arr'$, that $D' \subseteq D$. 
      
      In order to prove algebraic length 1, we take the ranking $\arr_{01}$ and regard $x_0 \ \epsilon_0 \ x_1 \ \dots \ x_n$ as a ranked formula over $\arr_{01}$. The ranking $r$ of the variables can be chosen so that $r(x_0)=0$ and $r(x_n)=1$ since the walk has algebraic length 1. By applying the ranked homomorphism $(g^i)_{i \in \Int} $ mapping $(A; \arr_{01})$ to $(A; \arr'_{01})$ (used already in \Cref{lem:pp_from_rpp}) to the above walk, we obtain a walk in $\relstr{A}'$ of algebraic length 1 from $g^{r(x_0)}(a_0) = a_0$ to $g^{r(x_n)}(a_n)=g(a_n)=a_0$, as required.

Now $\relstr{B} = \relstr{A}'|_{D'}$ is finite, smooth,  $k$-linked for some $k$ (see the final paragraph of \Cref{subsec:prelim_digraphs}), and has no loops, since a loop $a \arr' a$ in $\relstr{A}'$ means $a \arr g^{-1}(a)$ in $\relstr{A}$ -- a loop in $\relstr{A}/\group G$. Moreover, recall from \Cref{subsec:ranked} that $D'$ is rpp-definable with parameters from $\arr_{01}$. 

If $\kpower = (D')^2$,  \cite[Claim 3.11]{Cyclic}  shows that $\arr_{01}$ rpp-defines with parameters a proper subset of $D'$ on which $\relstr{A}_0$ has a smooth weak component of algebraic length 1. Such a component, call it $D''$, is rpp-definable with parameters from $\arr_{01}$. The digraph $\relstr{A}'|_{D''}$ is finite, smooth, $k$-linked and without loops. We replace $\relstr{B}$ by this strictly smaller digraph  and continue.
      
If, on the other hand, $\kpower \neq (D')^2$,  we  apply \Cref{thm:main_finite} to $\relstr{B}$ and the trivial ranked automorphism group $\group H$ containing only $(\id)_{i\in\Int}$. The orbits of the projection of $\group H$ are pp-definable from the singleton unary relations, so the theorem shows that $\arr_{01}$ rpp-defines with parameters either $D'' \varsubsetneq D'$  such that $\relstr{A}'|_{D''}$ is still finite, smooth, $k$-linked (and without loops), or it rpp-defines 0-ranked $\OR(\alpha,\alpha)$ for a proper equivalence $\alpha$ on a subset of $A$. In the former case, we replace $\relstr{B}$ by $\relstr{A}'|_{D''}$ and continue.
In the latter case,  \Cref{lem:pp_from_rpp} finishes the proof. 
\end{proof}

\section{Missing proofs from \Cref{sec:main_finite_proof}}

\subsection{Obtaining P-central or PQ-central relations}\label{Pinsker}
  
   This section is devoted to the proof of~\Cref{prop:gettingTSR}
  and~\Cref{prop:Rosenberg}.
    Both proofs are based on a single construction.
    We take a relation $R\subseteq A^n$ on a finite $A$ and define
    \begin{equation*}
      \upRel {R}{k}(x_1,\dotsc,x_{k}) = \exists y\ \bigwedge_{1\leq i_1<\dotsc <i_{n-1}\leq k} R(y,x_{i_1},\dotsc,x_{i_{n-1}}).
    \end{equation*}
    We will state several properties of the constructions, their proofs are straightforward and are omitted.
    First we note that if $R$ is totally symmetric, then so is $\upRel R k$.
    
    For  $n=2$  we have the following properties.
    \begin{itemize}
      \item $\upRel R{k}$ is ``have common in-neighbour''. 
      \item If $R$ is linked, then so is $\upRel R 2$.
      \item $\upRel R k$ is totally symmetric,
      \item If $R$ is subdirect, then $\upRel R 2$ is reflexive~%
        (which for arity 2 means the same as totally reflexive),
      \item If $\upRel R k$ is full, then $\upRel R {k+1}$ is totally reflexive.
      \item For any ranking of $R$, $\upRel R k$ is defined by a $0$-ranked formula in $R$.
    \end{itemize}
    If $R$ is totally reflexive of arbitrary arity $n$, then 
    \begin{itemize} 
      \item $\upRel R{n-1} = A^{n-1}$,
      \item $R\subseteq \upRel R n$ (use the first entry of a tuple as a witness for $y$), and 
      \item if $\upRel R{k} = A^k$, then $\upRel R{k+1}$ is totally reflexive.
    \end{itemize}

    \gettingTSR*

    \begin{proof}
      Proof is a direct application of the construction from previous subsection with $n=2$. Let $R$ be binary, subdirect, linked, but not central.  
      We find the smallest $k\geq 2$ such that $\upRel R {k}\neq A^{k}$
      and take $\upRel R {k}$ to be our TSR relation.
      By the properties above, we have that this is indeed a TSR relation and that the pp-formula is $0$-ranked.
      Moreover, if $k=2$, then $\upRel R 2$ is linked. 
      
      We need to show that $\upRel R {k}\neq A^{k}$ for some $k$.
      Indeed, $\upRel R {|A|} \neq A^{|A|}$,
      otherwise we would have $(a_1,\dotsc,a_{|A|})\in \upRel R {|A|}$~%
      (where the tuple lists all the elements of $A$).
      Therefore all the vertices have a common in-neighbour and $R$ is central,  a contradiction.
    \end{proof}
    
 \PorPQ*

    \begin{proof}
      Let $R$ be a proper TSR relation of arity $n$, which is additionally linked if $n=2$.
      In order to prove the statement, we repeatedly apply the construction from the previous subsection.
      In each step, if $\upRel R {n} = R$, we stop.
      Otherwise, we find the smallest $k\geq n$ such that $\upRel R {k}\neq A^{k}$;
      if there is no such $k$, we stop.
      In the remaining case, we substitute $\upRel R {k}$ for $R$ (thus  $n$ changes to $k$)
      and repeat.

      Note first that the starting $R$ is a proper TSR relation and so is the new $R$ after each replacement. This follows from the properties of the construction in the previous subsection. 

      The next thing to note is that the process must stop.
      Indeed,  if $k=n$, then we add at least one tuple to the relation,
      and otherwise we increase the arity.
      Since $A$ is finite, we must stop,
      as every totally reflexive relation over $A$ of arity $>|A|$ is full. 

      Let's say we stopped with $R$ of arity $n$. 
      We first consider the case that each $\upRel R k$ is full. 
      Then $\upRel R {|A|}$ is full, and in particular includes the tuple $(a_1,\dotsc,a_{|A|})$ 
      enumerating the elements of $A$.
      The witness~%
      (the value of $y$ from the pp-definition of $\upRel R k$),
      say $a$, for this tuple is in the P-center of $R$ and thus $R$ is P-central.
      Indeed, a tuple $(a,a_1,\dotsc,a_{n-1})$ is in $R$ whenever  $a_i=a_j$ for $i\neq j$, since $R$ is totally reflexive.
      If all the $a_i$ are different, the tuple is in $R$, since $a$ was a witness and $R$ is totally symmetric.

      In the remaining case, we have $\upRel R n = R$.
      If $n=2$, then $R$ is an equivalence. But
      this cannot happen, since the original relation was linked and proper and the iterative process keeps these properties until it increases the arity.

      Now we have $n>2$. We show $R$ is PQ-central.
      Let $\alpha$ be defined from $R$ as in~\Cref{def:2center}.
      Since $R$ is totally symmetric, $\alpha$ is symmetric,
      and since $R$ is totally reflexive, $\alpha$ is reflexive.
      It remains to prove that $\alpha$ is transitive.
      Take $(a,b),(b,c)\in\alpha$. 
      We will show that $(a,c,a_1,\dotsc,a_{n-2})\in\upRel R n$ for any $a_1,\dotsc,a_{n-2}$;
      since $\upRel R n = R$, this will conclude the proof.
      We take the value for $y$ to be $b$. By total symmetry, it is now enough to show 
      that if we switch one element $(a,c,a_1,\dotsc,a_{n-2})$ for $b$, the resulting tuple is in $R$.
      But this is obvious as the resulting tuple always either has  both $b$ and $c$, or both $b$ and $a$, 
      and $(a,b),(b,c)\in\alpha$.
    \end{proof}

\subsection{Second step, central case} \label{apps:step_two_b}

      \LemmaWalking*
      
      \begin{proof} 
      We first observe that from any $B' \subseteq A$ we can rpp-define the largest subset $B \subseteq B'$ such that $\relstr{A}|_B$ is smooth, the \emph{smooth part of $B$}. Indeed, $B$ consists of all the elements that have an infinite outgoing and an infinite incoming walk in $B'$. Since $B'$ is finite, a walk of length $|B'|$ suffices, and we can use the following ranked formula with free variable $x$. 
      \begin{align*}
        &\exists x_1,\dotsc,x_{|B'|},x'_1,\dotsc,x'_{|B'|}\\ &B'(x)\wedge\bigwedge_i B'(x_i)\wedge\bigwedge_i B'(x'_i) \wedge \\
        &x_1\rightarrow x_2\rightarrow\dotsb\rightarrow x_{|B'|} \rightarrow x\rightarrow x'_1\rightarrow\dotsb\rightarrow x'_{|B'|}
      \end{align*}
      
      It is therefore enough to rpp-define a proper $B'$ with nonempty smooth part. 
        Let $\varphi$ be a ranked formula in $\arr$  that defines the $k$-linkness equivalence for $|A|$-element digraphs (recall~\Cref{subsec:prelim_digraphs,subsec:ranked}). It is equal to $A^2$ as $\relstr{A}$ is linked.
        We name the variables $x_0$, \dots, $x_n$ so that $x_0$ and $x_n$ are the free variables and, for each $i$, $x_i \arr x_{i+1}$ or $x_{i+1} \arr x_i$ is a conjunct. Let $\psi_i$ be the ranked formula obtained from $\varphi$ by keeping only the variables $x_0$, $x_1$, \dots, $x_i$ (removing the other conjuncts), making $x_i$ free, and adding the conjunct $C(x_0)$. Let $C_i$ be the subset of $A$ defined by $\psi_i$.
        Clearly, $C_0=C$ and $C_n=A$.  It follows that, for some $i$, $\plusarr{C_i}=C_{i+1}=A$ or $\minusarr{C_i}=C_{i+1}=A$, while $C_i$ is proper. %
        
        We set $B'=C_i$ and observe that the smooth part of $B'$ is nonempty. Indeed, if $\plusarr{(B')}=A$, then there is an infinite backward walk in $B'$, and this implies at least one directed cycle in $\relstr{A}|_{B'}$. Each element of this cycle is in the smooth part.
        Noting that the case $\minusarr{(B')}=A$ is analogical, the proof is concluded.
      \end{proof}

      \StepTwoB*
      
     \begin{proof} 
        Let $C$ be the center of $R$ and $\psi$ be the rpp-formula defining $B$ from \Cref{lem:walking}, i.e., $B$ is a proper nonempty subset of $A$ and $\relstr{A}|_B$ is smooth.
        We assume that $\relstr{A}|_B$ is not $k$-linked; let $\alpha$ be the $k$-linkness equivalence relation on $B$. Our aim is to rpp-define $\OR(T,T)$ for a proper TSR $T$. 
        
        Take a 0-ranked rpp-formula $\varphi$ with two free variables defining the $k$-linkness relation. %
        Since $\relstr{A}$ is linked, it defines $A^2$. 
        Let $\varphi'$  be obtained from $\varphi$ by adding a conjunct $B(x)$ for every variable~%
        (i.e., both the quantified and the free variables).
        Now $\varphi'$ means ``being $k$-linked in $\rightarrow$ restricted to $B$'', so it is a $0$-ranked pp-definition of $\alpha$.

        Next, we define $\varphi''$ by replacing in $\varphi'$  each conjunct $B(x)$ on a quantified variable $x$ by the formula $\psi(x)$. Clearly, the formula $\varphi''$ still defines  $\alpha$. Also note that if we remove all the conjuncts $C(x)$ (they all come from quantified variables), then all the restrictions on the original quantified variables of $\varphi'$ are dropped (see the remark before the lemma), so the obtained formula defines $B^2$. 
        
        In $\varphi''$ we remove, one by one, the conjuncts $C(x)$. %
        At some point we arrive to a formula with a selected, quantified variable $x$ that defines a subset of $B^2$ strictly larger than $\alpha$, but if we added back the conjunct $C(x)$, it would define $\alpha$. By making $x$ free, we get an rpp-definition (using $B$, $C$, $\arr_{01}$) of a ternary relation $S$ such that   

        \begin{itemize}
          \item $\exists x\ S(y,y',x)$ is a subset of $B^2$ strictly larger than $\alpha$, and
          \item $\exists x\ S(y,y',x)\wedge C(x)$ is $\alpha$.
        \end{itemize}
        By shifting the ranking $r$ if necessary, we can assume that $r(1)=r(2)=0$, while  $r(3)$ can be different.%

        The remaining reasoning is somewhat similar to  what was done in \Cref{lem:orfromQ}. We choose $a,b\in B$ such that:
        \begin{itemize}
          \item $(a,b)\notin \alpha$ and
          \item there exists $d$ such that $S(a,b,d)$ and $I = d+R$ is maximal~(under inclusion) among similar sets defined for other $(a,b)\notin\alpha$ and $d$. %
        \end{itemize}
        Note that $I\neq A$ as $S(a,b,d)$ for $(a,b)\notin\alpha$ requires $d$ outside of $C$ by the second property of $S$.
        We assume that $A = \{1, 2, \dots, n\}$, $a=1$, and $b=2$.

    We choose $k$ and $T$ exactly as was done for PQ-centers. The number $k$ is such that  every $(k-1)$-element subset of $A$ is included in $g(I)$ for some $g \in \group G$, and  some $k$-element subset of $A$ is not included in $g(I)$ for any $g \in \group G$.
        Our $T$ will consists of tuples $(c_1,\dotsc,c_k)$ such that
        $\{c_1,\dotsc,c_k\}\subseteq g(I)$ for some $g$ from $\group{G}$. The relation $T$ is TSR and $T \neq A^k$. 
    
        Our rpp-formula with free variables $x_1^1,\dotsc,x_k^1, x_1^2,\dotsc,x_k^2$ is defined as follows.
        \begin{align*}
          &\exists y_0,y_1,y_2,X_1,X_2,z_1,\dotsc,z_n,z_1',\dotsc,z_n'\\  
          &y_0 = z_1 \wedge y_2 = z_2 \wedge
          O(z_1,\dotsc,z_n) \wedge O(z_1',\dotsc,x_n') \wedge\\
          &\bigwedge_{j=1}^2 \big( S(y_{j-1},y_j,X_j) \wedge \bigwedge_{i=1}^k R(X_j,x_i^j) \wedge 
          \bigwedge_{i\in I} R(X_j,z'_i) \enspace.
        \end{align*}
        By shifting the ranking $r$ if necessary, we have  $r(y_i) = r(z_i) = 0$ and $r(x_i^j) = r(z_i') =l$~%
        (for some suitably chosen $l$).

        It remains to verify that the formula works.
        First take $(a_1,\dotsc,a_k,b_1,\dotsc,b_k)$
        such that all entries of $(a_1,\dotsc,a_k)$ are in $g(I)$ for some $g \in \group G$~%
        (again, the case of $b_i$'s completely analogous).
        Take $\{g_i\}_i\in\group H$ such that $g_l=g$.

        We evaluate $y_0 \mapsto g_0(1)$, $y_1,y_2\mapsto g_0(2)$,  $X_1\mapsto g_{r(X_1)}(d)$, $X_2$ to $d'$, an element of $C$  
        such that $S(g_0(2),g_0(2),d')$ holds.
        Finally, we evaluate $z_i\mapsto g_0(i)$ and $z'_i\mapsto g_l(i)$. 
        
        The first four conjuncts hold by construction.
        Let's focus on the last, complex conjunct. 
        For $j=2$, we evaluated $X_2$ as $d'$ so that the $S$-conjunct holds and $d'$ is in the center of $R$, and thus the whole conjunct holds for any $b_i$'s and any values of $z'_i$.
        For $j=1$, 
        note that $S(1,2,d) \wedge R(d,i)$ and thus $S(g_0(1),g_0(2),f_{r(X_1)}(d))\wedge R(g_{r(X_1)}(d),g_l(i))$ 
        for every $i\in I$ and therefore the conjunct holds as well.%

        For the opposite direction, let $\val$ denote a satisfying evaluation of variables. 
        Define $g_l:i\mapsto z'_i$ and extend it to $\{g_i\}_i\in\group H$.
        A new evaluation of variables is obtained by composing the old one with the inverse of $\{g_i\}_i$: $\val'(x) = g_{r(x)}^{-1}(\val(x))$ 
        The map $\val'$ satisfies all the conjuncts in a formula and,
        additionally, $\val'(z'_i)=i$ for all $1\leq i \leq n$.

        Note that $(\val'(y_0),\val'(y_2))\notin\alpha$.
        Indeed we have $O(\val'(y_0),\val'(y_2),\val'(z_3),\dotsc,\val'(z_n))$, i.e., there is an element of $\group G$
        such that $1\mapsto \val'(y_0)$ and $2\mapsto \val'(y_2)$.
        An inverse of this function is also in $\group G$, but then $(\val'(y_0),\val'(y_2))\in\alpha$
         would imply $(1,2)\in\alpha$, and we know this is not the case.

        Therefore for $j$ equal to $1$ or $2$ we have $(\val'(y_{j-1}),\val'(y_j))\notin\alpha$ and since
        \begin{equation*}
          S(\val'(y_{j-1}),\val'(y_j),\val'(X_j)) \wedge R(\val'(X_j),\val'(z_i')=i) 
        \end{equation*}
        holds for every $i\in I$ we conclude that
        \begin{itemize}
          \item $\val'(X_j)$ is not in $C$, and 
          \item $\val'(x_i^j)\in I$ for every $i\in\{1,\dotsc,k\}$, as otherwise 
            $(\val'(y_{j-1}),\val'(y_j))$ together with $\val'(X_j)$ would be a strictly better choice than $(a,b)$ and $d$.
        \end{itemize}
        As $\val(x_i^j) = g_l(\val'(x_i^j))$ we proved the second inclusion and the lemma holds.
      \end{proof}

\subsection{Improving OR} \label{apps:step_three}

     We start by proving \ref{prop:gettingPorPQor}.
     The proof is based on the following auxiliary result.

      \begin{lemma}\label{prop:changeOR}
        Let $R,S,R'$ be relations such that $R'$ has an equality-free pp-definition in $R$.
        Then $\OR(R,S)$ pp-defines $\OR(R',S)$.
      \end{lemma}
      \begin{proof} %
        Let $\phi$ be an equality free formula defining $R'$ from $R$.
        To obtain a pp-definition of $\OR(R',S)$ we rewrite $\phi$. 
        Say the arity of $R$ is $n$, of $R'$ is $n'$, and of $S$ is $k$. We start with $\phi$ and then
        \begin{itemize}
          \item add $k$ fresh free variables $y_i$, and
          \item substite each occurrence of $R(z_1,\dotsc,z_n)$ with $\OR(R,S)(z_1,\dotsc,z_n,y_1,\dotsc,y_k)$
        \end{itemize}
        Let the new formula define relation $T$. We show $T = \OR(R',S)$ by verifying the two inclusions.

        If $T(a_1,\dots,a_{n'},b_1,\dotsc,b_k)$,  we have two cases.
        If $S(b_1,\dotsc,b_k)$, then we are done~%
        (as the longer tuple is certainly in $\OR(R,S)$).
        On the other hand, if it is not, then every occurrence of 
        $\OR(R,S)(z_1,\dotsc,z_{n},y_1,\dotsc,y_k)$
        forces $R(z_1,\dotsc,z_n)$ and we get $\OR(R',S)(a_1,\dotsc,b_k)$ due to $R'(a_1,\dotsc,a_{n'})$.

        Let $\OR(R',S)(a_1,\dotsc,a_{n'},b_1\dotsc,b_k)$.
        If this happens due to $R'(a_1,\dotsc,a_{n'})$, we evaluate quantified variables according to $\phi$ defining $R'$ from $R$. Then
         every conjunct of our new formula holds due to $R$. 
        If, on the other hand, $\OR(R',S)(a_1,\dotsc,a_{n'},b_1,\dotsc,b_k)$ holds due to $S(b_1,\dotsc,b_k)$,
        then every conjunct of the new relation holds no matter what values the first $n$ variables take. %
      \end{proof}

       \GettingPorPQor*

      \begin{proof}
        Let $n$ be the arity of $T$;
        if $n=1$ we already have $\OR(S,S)$ for unary $S=T$.

        If $n=2$ we have two cases: we know that $T$ is reflexive and symmetric,
        but we don't know if it is linked.
        If it is not linked we let $\alpha$ to be the transitive closure of $T$;
        clearly $\alpha$ is an equivalence on $A$ and can be defined without equality from $T$.
        We apply~\Cref{prop:changeOR} twice~%
        (once to the first $T$ and then, after permuting coordinates, to the second $T$)
        to get $\OR(\alpha,\alpha)$ --- we are done.
  
        In the final case $n=2$ and $T$ is linked or $n>2$;
        in this case the equality free formula is provided by~\Cref{prop:Rosenberg},
        and applying~\Cref{prop:changeOR} twice we get $\OR(S,S)$ for a TSR P-central or TSR PQ-central $S$.
      \end{proof}

   Before the final lemma of the proof, we observe   that $\OR(C,C)$, where $C$ is unary,  pp-defines $C$ (by the formula $C(x,x)$) and $\OR(C^k,C^k)$ for any $k$ by the formula
      \begin{multline*}
        \OR(C^k,C^k)(x_1,\dotsc,x_k,y_1,\dotsc,y_k)\equiv \\
        \bigwedge_{i=1}^k \bigwedge_{j=1}^k \OR(C,C)(x_i,y_j).
      \end{multline*}

\FinalOR*

      \begin{proof}
        We begin as in the proof of \Cref{lem:orfromcenter}. We take a rpp-formula $\psi$,
        which is provided by~\Cref{lem:walking},
        and uses $C$~%
        (defined by $\OR(C,C)(x,x)$)
        and $\arr_{01}$  
        to define 
        $\emptyset \neq B \varsubsetneq A$ such that $\relstr{A}|_B$ is smooth.
        We assume that $\relstr{A}|_B$ is not $k$-linked and denote $\alpha'$ be the $k$-linkness equivalence relation on $B$. 
        
        We define the formula $\varphi''$ as before and start releasing the conjuncts $C(x)$ for quantified $x$. This time, however, with a different criterion for termination: we stop when our binary relation becomes linked. By making the key variable free, we get a ternary relation $S$ such that
        \begin{itemize}
          \item $\exists X\ S(y,y',X) \wedge C(X)$ is contained in a non-trivial equivalence on $B$, and
          \item $\exists X\ S(y,y',X)$ is not contained in any proper equivalence on $B$.
        \end{itemize}
        The new formula is ranked with the first two variables ranked $0$ (after shifting if necessary). %

        Both binary relations $\exists X\ S(y,y',X) \wedge C(X)$ and $\exists X\ S(y,y',X)$ are %
     reflexive on $B$ as both contain $\alpha'$ and
        the smallest equivalence containing a reflexive relation is pp-definable from it.
        In fact, we can find a formula in a binary symbol $Q$ which defines the symmetric and transitive closure for both of the binary relations at the same time~%
        (depending on the interpretation of the symbol $Q$).
        We replace each binary conjunct of $Q(y,y')$ by $S(y,y',X_i)$, with a fresh $X_i$, and obtain an rpp-definition of a relation $U$ such that 
        $\exists X_1,\dotsc,X_p\ U(y,y',X_1,\dotsc,X_p)$ defines $B^2$,
        while $\exists X_1,\dotsc,X_p\ U(y,y',X_1,\dotsc,X_p) \wedge \bigwedge_i C(X_i)$ defines a proper equivalence on $B$, call it $\alpha$.
        The rank of the first two variables in $U$ is $0$ and all the $X_i$  have equal rank.

        Recall that $\OR(C,C)$ pp-defines $\OR(C^p,C^p)$.
        Now we have 
        \begin{align*}
          \OR(\alpha,\alpha)\equiv
          &\exists X_1,\dotsc,X_p, Y_1,\dotsc,Y_p \\
          &\OR(C^p,C^p)(X_1,\dotsc,X_p,Y_1,\dotsc,Y_p) \wedge \\
           &U(y,y',X_1,\dotsc,X_p) \wedge U(z,z',Y_1\dotsc,Y_p).
        \end{align*}
        and the relation is $0$-ranked --- we are done.
      \end{proof}

\section{Missing proofs from \Cref{sec:infinite}}

\subsection{Infinite pseudoloop lemma}  \label{apps:one_b_refined}

\OneBRefined*

Let $\relstr{A}$ and $\group{G}$ be as in the statement and assume additionally that $\relstr{A}/\group{G}$ has no loop. 
We will show that $\arr$ and orbits of $\group{G}$ 1-dimensionally pp-interpret digraph $\relstr{B}$ such that $\relstr{B}$ with the natural action, call it $\group{H}$, of $\group{G}$ on the universe satisfy the following conditions.
\begin{itemize}
    \item $\relstr{B}$ is smooth,
    \item $\relstr{B}/\group{H}$ has algebraic length 1,
    \item $\relstr{B}/\group{H}$ has no loop, and
    \item all weak components of $\relstr{B}$ are finite.
\end{itemize}
In this situation, \Cref{thm:one_a_refined} finishes the proof.

To construct $\relstr{B}$, 
we start by applying~\Cref{thm:triangle_config} and obtain $k$ and a triangle configuration $(P,P_0,P_1,P_2)$ for $(A;\kpower)$ and $\group G$.
Let $\sim$ be a maximal equivalence relation on $P$ which is
pp-definable from $\kpower$ and the four sets and which is contained in  $P_0^2 \cup P_1^2 \cup P_2^2$ (the equivalence with equivalence classes $P_0$, $P_1$, $P_2$; this one, unfortunately, need not be pp-definable). Such an equivalence relation exists since every (pp)-definable equivalence relation is a union of 2-orbits of $\group{G}$, and by our assumption $\fa$ has finitely many 2-orbits. 
We set $\relstr{B} = (\relstr{A}|_P)/{\!\sim}$ and observe that $\relstr{B}$ is indeed smooth (as $\plusarr{P}$ and $\minusarr{P}$ both contain $P$), $\relstr{B}/\group{H}$ has algebraic length 1 (as $\relstr{A}|_P/\group{G}$ is non-bipartite), and $\relstr{B}/\group{H}$ has no loop (as $\sim$ is contained in $P_0^2 \cup P_1^2 \cup P_2^2$ and the $P_i$ are independent).
It is therefore enough to verify the fourth condition, that all weak components of $\relstr{B}$ are finite.

In order to do that, we will first make several simplifying assumptions.
First, we can clearly assume that $k=1$ (by simply renaming the relation), that $A=P$ (by working with $\relstr{A}|_P$ instead of $\relstr{A}$), and that $\sim$ is the trivial equivalence (by working with $\relstr{A}/{\!\sim}$ instead of $\relstr{A}$). Now we are in the following position.
\begin{itemize}
\item $(A;\arr)$ is smooth, $A$ is the disjoint union of the $P_i$,
\item $\group{H}$ is a subgroup of $\aut(A;\arr,P_1,P_2,P_3)$ that has finitely many 2-orbits,
\item $\plusarr{P_i} = P_{i-1} \cup P_{i+1}$ for each $i$, and
\item there is no nontrivial equivalence relation $\sim$ which is pp-definable from the $P_i$ and $\arr$, and which is contained in $P_0^2\cup P_1^2\cup P_2^2$.
\end{itemize}

The crucial fact is the following. 

\begin{lemma} \label{lem:asdf}
	If $v,v' \in P_i$ for some $i$, and they have a common in-neighbor or a common out-neighbor, then $v=v'$.
\end{lemma}
	
	\begin{proof}
	We assume that $v,v' \in P_1$ and $u \in P_0$ are such that $v\leftarrow u\arr v'$, and we aim to show that $v=v'$. The cases of  different positions and common out-neighbor are proved analogically.
	
	We write $\overline{P_i}$ for $A\setminus P_i$. Note that $\overline{P_i}=P_{i-1} \cup P_{i+1} = P_i^+$, in particular, the subsets $\overline{P_i}$ are pp-definable from $\arr$ and $P_i$.

	We consider the following binary relations $S_j$, $j = 1,2,3$. 
\begin{align*}
S_j(x,y) \equiv \exists& z_1,z_2,z_3 \ \ x\leftarrow z_1\arr z_2\leftarrow z_3\arr y\\ & \wedge \ z_1\in \overline{P_j}\wedge z_2\in \overline{P_j}\wedge z_3\in \overline{P_{j+1}} \enspace.
\end{align*}

	We claim that for all $j$ the following hold.
	
\begin{enumerate}[label=(\alph*)]
\item The relation $S_j$ is reflexive.
\item $S_j(v,v')$ and $S_j(v',v)$.
\item $S_j\cap (P_{j-1}\times P_j)=\emptyset$.
\end{enumerate}

	To prove item (a), let $a\in A$ be arbitrary. The evaluation witnessing $S_j(a,a)$ is found as follows. 

\begin{itemize}
\item If $a\in P_j$, we pick $z_1\in P_{j-1}$ and $z_2\in P_{j+1}$ such that $a\leftarrow z_1\arr z_2$, and we put $z_3 = z_1$. The choices are possible by $\plusarr{P_{j-1}} \supseteq P_j \ni a$ and $\minusarr{P_{j+1}} \supseteq P_{j-1} \ni z_1$.
\item If $a\in P_{j+1}$ we pick $z_1\in P_{j-1}$ such that $a\leftarrow z_1$, and we put $z_2 = a$ and $z_3 = z_1$.
\item If $a\in P_{j-1}$ we pick $z_1\in P_{j+1}$ and $z_3\in P_j$ such that $a\leftarrow z_1\arr a$ and $a\leftarrow z_3\arr a$, and we put $z_2=a$.
\end{itemize}

	In order to prove item (b), first observe that the role of $v$ and $v'$ is symmetric, so it is enough to show that $S_j(v,v')$. This time, we distinguish 2 cases.
	
\begin{itemize}
\item If $j=0$ or $j=2$, then we pick $z_1=u, z_2=v'$ and $z_3\in P_2$ such that $v'\leftarrow z_3$.
\item If $j=1$, then we pick $z_1=z_3=u$ and $z_2\in P_2$ such that $u\arr z_2$.
\end{itemize}

	In order to prove (c), let $a\in P_{j-1}$ and  $S_j(a,b)$. We will show that $b \not\in P_j$. Consider again an evaluation witnessing $S_j(a,b)$. Then, since $a \larr z_1$ and $z_1 \in \overline{P_j}$, we have $z_1\in \overline{P_{j-1}}\cap \overline{P_j}=P_{j+1}$. Continuing similarly, we obtain that $z_2\in P_{j-1}$, then $z_3\in P_j$, and finally $b\in \overline{P_j}$. 
	
	\smallskip
	
	Now, when properties (a), (b), and (c) are verified, we can finish the proof as follows.
	We define $S =  \bigcap_{j}S_j\cap \bigcap_{j}S_j^{-1}$. By definition, $S$ is symmetric. From (b) it follows that $S(v,v')$, and from (c) it follows that  $S\subseteq P_0^2\cup P_1^2\cup P_2^2$. It is also clear from our construction that $S$ is pp-definable from $\arr$ and the $P_i$. Let $\sim$ be the transitive closure of $S$. The reflexivity of $S$ implies that   $S\subseteq 2S \subseteq 3S \subseteq \dots$. On the other hand, all $kS$ are invariant under $\group H$, so they are unions of 2-orbits of $\group H$. Since $\group H$  has finitely many 2-orbits, we have  $\sim=S^{\circ k}$ for some finite number $k$. In particular, $\sim$ is  pp-definable from $\arr$ and the $P_i$. There are no such nontrivial equivalences, therefore $\sim$ is trivial, whence  $v=v'$, which finishes the proof.
	\end{proof}
	
	We are ready to finish the proof of \Cref{thm:one_b_refined} by showing that  every weak component of $(A;\arr)$ is finite.  Let $C$ be a weak component of $(A,\arr)$, and let $c\in C$ be arbitrary. We claim that, for all $a,b\in C$, if the pairs $(c,a)$ and $(c,b)$ are in the same orbit of $\group H$, then $a=b$. Since $\group H$ has only finitely many 2-orbits, this fact implies that $C$ must be finite. Let $\alpha\in \group H$ be such that $\alpha(c)=c$ and $\alpha(a)=b$. Let  $w$ be a walk $c = c_0\ \epsilon_0\ \dots \  \epsilon_{n-1} \ c_n=a$ from $c$ to $a$. Then $\alpha(w')=\alpha(c_0) \ \epsilon_0 \ \dots \ \epsilon_{n-1}\ \alpha(c_n)$ is an oriented walk from $c$ to $b$. Since $P_i$ are invariant under $\group{G}$, each $c_i$ is in the same set $P_{j_i}$ as $\alpha(c_i)$. Lemma~\ref{lem:asdf} now gives us $c_1 = \alpha(c_1)$, $c_2 = \alpha(c_2)$, and, eventually, $a=c_n=\alpha(c_n)=b$, and the proof is concluded,

\subsection{Legs are adjacent} \label{apps:cutoff}

We prove here one of the lemmata in the proof of \Cref{thm:triangle_config} from \Cref{subsec:triangle_config}.

cutoff*

\begin{proof}
	We prove $\psarr{V_1} \supseteq V_0$, there is no difference for other $i,j$. That is, we will show that each $v \in V_0$ has a neighbor in $V_1$.

	We denote $P_{i} =  \bigcup{U_i}$ and 
	define a subset $B \subseteq A$ in $\relstr{A}$ by the following formula. 
\begin{align*}
B(x) \equiv  \exists& y_1,y_2,y_3 \ \ x\arr y_1\leftarrow y_2\arr y_3\leftarrow x\\ &\wedge \plusarr{P_1}(y_1) \wedge \plusarr{P_1}(y_2) \wedge  \plusarr{P_0}(y_3).
\end{align*}
	Note that the sets $P_i$ are pp-definable from $\arr$ and 1-orbits $\fa$. Therefore, so is $B$.
	
	Recall that there exists a triangle in $\OG$ with vertices $u_i \in U_i$.
	We claim that $B$ contains all elements of the orbits of $u_i \in U_i$.
	Indeed, for $x\in u_0\cup u_1$, a witnessing evaluation is  $y_1=y_3 \in \plusarr{x}\cap u_2$ and $y_2\in \minusarr{y_1}\cap u_0$; for $x\in u_2$ a witnessing evaluation is $y_1\in \plusarr{x}\cap u_0$, $y_2=x$, and $y_3\in \plusarr{x}\cap u_1$. The same argument as in~\Cref{lem:tree_def_bipartite} gives us $B=A$.

	Now let $x\in v \in V_0$ be arbitrary. Since $B=A$, we know that $B(x)$. Take  $y_1,y_2,y_3\in A$ witnessing this and let $v_1,v_2,v_3 \in U$ be their $\group G$-orbits. We have $v \sarr v_1 \sarr v_2 \sarr v_3 \sarr v$, $v_1,v_2 \in \psarr{U_1}$, and $v_3 \in \psarr{U_0}$. We claim that $v_1\in V_1$. 
	If not, then $v_1 \in U_0$ as $v_1 \in \psarr{U_1} = U_0 \cup U_2 \cup V_1$ and it is a neighbor of $v \in V_0$. 
	Then $v_2 \in U_2$ as $v_2 \in \psarr{U_1}$ and it is a neighbor of $v_1 \in U_0$. Finally, $v_3 \in U_1$ as $v_2 \in \psarr{U_0}$ and it is a neighbor of $v_2 \in U_2$. But then $v$ and $v_3$ cannot be neighbors, a contradiction. 
\end{proof}

\subsection{Weakest pseudoloop conditions} \label{apps:pseudo_conditions}

\ThmPseudoConditions*

	We first formulate a general statement which guarantees the satisfaction of a pseudoloop condition in a model-complete core.
	
\begin{lemma}\label{use_loops}
	Let $\relstr A$ be an $\omega$-categorical model-complete core, and let $\relstr C$ be a digraph such that the following holds. 

	For all digraphs $\relstr B=(B,\arr)$ and groups $\group G\leq \aut(\relstr B)$ with finitely many 1-orbits such that $\relstr B/{\group G}$ contains a (not necessarily induced) subgraph isomorphic to $\relstr C$ and $\relstr A$ pp-interprets $\relstr B$ together with the 1-orbits of $\group G$, $\relstr B$ has a pseudoloop modulo $\group G$.
	
	Then $\relstr A$ satisfies the $\relstr C$-pseudoloop condition.
\end{lemma}

	The proof of~\Cref{use_loops} is a straightforward generalization of the proof of Lemma 4.3 in~\cite{Topo}. For the convenience of the reader we provide the argument below.

\begin{proof}
	We show that under the assumptions of the lemma the $\relstr C$-pseudoloop condition is satisfied locally in $\relstr A$, that is for all finite subset $F$ of $A$ there exists $f\in \pol(\relstr A) $ such that $f|_{F^n}$ satisfies the $\relstr C$-pseudoloop condition. The fact that the local satisfaction of a pseudoloop condition implies the global satisfaction follows from a standard compactness argument; for the details we refer the reader to Lemma 10.1.10 in~\cite{Book}.
	
	Let $[m]$ be the vertex set of $\relstr C$, and let $(\sigma_1,\tau_1),\dots,(\sigma_n,\tau_n)$ be an enumeration of the edges of $\relstr C$.
	
	The clone $\pol(\relstr A)$ has natural action of functions $A^{F^m}$, i.e., the set of $m$-ary functions from $F$ to $A$. We define a graph $\relstr B$ on $A^{F^3}$ whose edges are generated by $(\pi_{\sigma_j}^m,\pi_{\tau_j}^m): 1\leq j\leq n$ via the action of $\pol(\relstr A)$, and we define $\mathcal{G}$ be the group of coordinatewise action of $\aut(\relstr A)$ on $A^{F^3}$. By definition $\pol(\relstr A)$ preserves the edge relation on $\relstr B$. Moreover, the 1-orbits of $\mathcal{G}$ correspond to $|F^3|$-orbits of $\aut(\relstr A)$, thus there are finitely many of them, and, since $\relstr A$ is a model-complete core, they are all preserved by $\pol(\relstr A)$. By our construction we also know that $\relstr B$ contains a subgraph isomorphic to $\relstr C$. Thus, by the hypothesis of the lemma it follows that $\relstr B$ contains a pseudoloop modulo $\group G$. This means exactly that the $\relstr C$-pseudoloop condition is satisfied by $\pol(\relstr A)$ on $F$.
\end{proof}

\begin{proof}[Proof of~\Cref{thm:pseudo_conditions}]
    We first show the implication 1)$\rightarrow$3). Note that if $\relstr C$ is a subgraph of $\relstr C'$ then the $\relstr C$-pseudoloop condition implies the $\relstr C'$-pseudoloop condition. Therefore, it is enough to show that $\relstr A$ satisfies all $C_{2k+1}$-pseudoloop conditions where $C_n$ denotes the cyclic graph on $[n]$ with edges $1\sarr 2\sarr \dots \sarr n\sarr 1$.
    
	We check the hypothesis of~\Cref{use_loops} for $\relstr C=C_{2k+1}$. Let $\relstr B$ and $\group G$ be as in hypothesis of~\Cref{use_loops}. Let 
\[
S_0\coloneqq B, S_{i+1}\coloneqq S_i\cap \plusarr{S_i}\cap \minusarr{S_i} 
\]

	Since $\relstr B/\group{G}$ is finite it follows that $S_n=S_{n+1}$ for some $n\in \mathbb{N}$. Then the digraph ${\relstr B}|_{S_n}$ is smooth and it still contains a subgraph isomorphic to $C_{2k+1}$. Moreover, $S_n$ is pp-definable in $\relstr B$, and the 1-orbits of ${\group G}|_{S_n}$ are all 1-orbits of $\group G$. This implies by our assumption that digraph ${\relstr B}|_{S_n}$ together with the 1-orbits of ${\group G}|_{S_n}$ does not pp-interpret $K_3$ with parameters. By Theorem~\ref{thm:one_b_refined} we obtain that ${\relstr B}|_{S_n}$ and thus also $\relstr B$ contains a pseudoloop modulo $\group G$. Therefore, the hypothesis of~\Cref{use_loops} is satisfied which shows that $\relstr A$ satisfies the $C_{2k+1}$-pseudoloop condition.

    The implication 3)$\rightarrow$2) is obvious.

    Finally, the implication 2)$\rightarrow$1) follows from the argument presented in the first paragraph of the proof of Theorem 4.5 in~\cite{Topo}.
\end{proof}

\section{Missing proofs from~\Cref{sect:no-pwnu}}
\label{app:pwnu}
\corecover*
\begin{proof}
	Let $e$ be an endomorphism of $\fa^{\otimes k}$, and let $B$ be the domain of $\fa^{\otimes k}$.
	Since $\neq$ is pp-definable in $\fa^{\otimes k}$, $e$ is injective.
	Since $e$ preserves $\sim$, it induces an endomorphism $e'$ of $\fa$ as follows: for all $t\in [B]^k$, one has $e'([t]_\sim)=[e(t)]_\sim$ (here $e$ is applied componentwise).
	Since $\fa$ is a core, $e'$ is an automorphism of $\fa$.
	If $(t^1,\dots,t^r)\not\in R^k$ for some relation $R$ of $\fa$, then $([t^1]_\sim,\dots,[t^r]_\sim)\not\in R$ by definition, so $(e'([t^1]_\sim),\dots,e'([t^r]_\sim))\not\in R$, i.e., $([e(t^1)]_\sim,\dots,[e(t^r)]_\sim)\not\in R$.
	Again by definition, this means that $(e(t^1),\dots,e(t^r))\not\in R^k$, so that $e$ preserves the complement of every relation of $\fa^{\otimes k}$ and is therefore an embedding.
	Since $\fa^{\otimes k}$ is homogeneous, the restriction of  $e$ to any finite set can be extended to an automorphism of $\fa^{\otimes k}$, so that $\fa^{\otimes k}$ is a model-complete core.
\end{proof}

\superposition*
\begin{proof}
	Let $\fb$ be the Fra\"iss\'e limit of the superposition of all the classes $\mathcal C(\fa_i,k_i)$.
	The proof is the same as the previous one, one defines structures $\fc_s$ using maps $(q^i_s)_{i\in\omega}$, and obtains an embedding $s'$ from $\fc$ to $\fb$ which turns out to be an injective polymorphism of $\fb$.
	The satisfaction of $\overline\Sigma$ is checked similarly.
\end{proof}

\end{document}